\documentclass[twoside,11pt]{article}
\usepackage{amsmath,amssymb,latexsym,graphicx,hyperref,epstopdf,subcaption,color,multirow}
\usepackage{algpseudocode,algorithm,epstopdf}
\usepackage[capitalise]{cleveref}
\normalsize
\newtheorem{theorem}{Theorem} 
\newtheorem{lemma}{Lemma} 
\newtheorem{corollary}{Corollary}
\newtheorem{proposition}{Proposition}

\newtheorem{definition}{Definition}
\newenvironment{proof}{{\sc Proof. }}{\hfill$\Box$\vspace{0.2in}}
\newcommand{\OPT}{\mbox{\sc OPT}}

\title{Approximation algorithms for maximally balanced connected graph partition}
\author{Yong~Chen\thanks{Department of Mathematics, Hangzhou Dianzi University.  Hangzhou, China.
	\texttt{\{chenyong,anzhang\}@hdu.edu.cn}}
\and
	Zhi-Zhong~Chen\thanks{Division of Information System Design, Tokyo Denki University.  Saitama, Japan.
	\texttt{zzchen@mail.dendai.ac.jp}}
\and
	Guohui~Lin\thanks{Department of Computing Science, University of Alberta.
	Edmonton, Alberta T6G 2E8, Canada.
	\texttt{guohui@ualberta.ca}}
	\thanks{Correspondence author.}
\and
	Yao~Xu\thanks{Department of Computer Science, Kettering University.  Flint, MI, USA.
	\texttt{yxu@kettering.edu}}
\and
	An~Zhang$^*$}%
\date{\today}
\begin{document}
\maketitle
\begin{abstract}
Given a simple connected graph $G = (V, E)$,
we seek to partition the vertex set $V$ into $k$ non-empty parts such that the subgraph induced by each part is connected,
and the partition is maximally balanced in the way that the maximum cardinality of these $k$ parts is minimized.
We refer this problem to as {\em min-max balanced connected graph partition} into $k$ parts and denote it as {\sc $k$-BGP}.
The general vertex-weighted version of this problem on trees has been studied since about four decades ago, which admits a linear time exact algorithm;
the vertex-weighted {\sc $2$-BGP} and {\sc $3$-BGP} admit a $5/4$-approximation and a $3/2$-approximation, respectively;
but no approximability result exists for {\sc $k$-BGP} when $k \ge 4$, except a trivial $k$-approximation.
In this paper, we present another $3/2$-approximation for our cardinality {\sc $3$-BGP} and then extend it to become a $k/2$-approximation for {\sc $k$-BGP},
for any constant $k \ge 3$.
Furthermore, for {\sc $4$-BGP}, we propose an improved $24/13$-approximation.
To these purposes, we have designed several local improvement operations, which could be useful for related graph partition problems. 

%
%
%

\end{abstract}

\section{Introduction}
We study the following graph partition problem:
given a connected graph $G = (V, E)$,
we want to partition the vertex set $V$ into $k$ non-empty parts denoted as $V_1, V_2$, $\ldots$, $V_k$
such that the subgraph $G[V_i]$ induced by each part $V_i$ is connected,
and the cardinalities (or called sizes) of these $k$ parts, $|V_1|, |V_2|, \ldots, |V_k|$,
are maximally balanced in the way that the maximum cardinality is minimized.
We call this problem as {\em min-max Balanced connected Graph $k$-Partition} and denote it as {\sc $k$-BGP} for short.
{\sc $k$-BGP} and several closely related problems with various applications
(in image processing, clustering, computational topology, information and library processing, to name a few)
have been investigated in the literature.

Dyer and Frieze~\cite{DF85} proved the NP-hardness for {\sc $k$-BGP} on bipartite graphs, for any fixed $k \ge 2$.
When the objective is to maximize the minimum cardinality, denoted as {\sc max-min $k$-BGP},
Chleb{\'{i}}kov{\'{a}}~\cite{Chl96} proved its NP-hardness on bipartite graphs (again),
and that for any $\epsilon > 0$ it is NP-hard to approximate the maximum within an absolute error guarantee of $|V|^{1 - \epsilon}$.
Chataigner et al.~\cite{CSW07} proved further the strong NP-hardness for {\sc max-min $k$-BGP} on $k$-connected graphs, for any fixed $k \ge 2$,
and that unless P $=$ NP, there is no $(1 + \epsilon)$-approximation algorithm for {\sc max-min $2$-BGP} problem, where $\epsilon \le 1/|V|^2$;
and they showed that when $k$ is part of the input, the problem, denoted as {\sc max-min BGP}, cannot be approximated within $6/5$ unless P $=$ NP.

When the vertices are non-negatively weighted, the weight of a part is the total weight of the vertices inside,
and the objective of vertex-weighted {\sc $k$-BGP} (vertex weighted {\sc max-min $k$-BGP}, respectively)
becomes to minimize the maximum (maximize the minimum, respectively) weight of the $k$ parts.
The vertex weighted {\sc $k$-BGP} problem is also called the {\em minimum spanning $k$-forest} problem in the literature.
Given a vertex-weighted connected graph $G = (V, E)$, a {\em spanning $k$-forest} is a collection of $k$ trees $T_1, T_2, \ldots, T_k$,
such that each tree is a subgraph of $G$ and every vertex of $V$ appears in exactly one tree.
The weight of the spanning $k$-forest $\{T_1, T_2, \ldots, T_k\}$ is defined as the maximum weight of the $k$ trees,
and the weight of the tree $T_i$ is measured as the total weight of the vertices in $T_i$.
The objective of this problem is to find a {\em minimum} weight spanning $k$-forest of $G$.
The equivalence between these two problems is seen by the fact that a spanning tree is trivial to compute for a connected graph.
The minimum spanning $k$-forest problem is defined on general graphs, but was studied only on trees in the literature~\cite{PS81,BSP82,Fre91,FZ17},
which admits an $O(|V|)$-time exact algorithm.

Not too many positive results from approximation algorithms perspective exist in the literature.
Chleb{\'{i}}kov{\'{a}}~\cite{Chl96} gave a tight $4/3$-approximation algorithm for the vertex-weighted {\sc max-min $2$-BGP} problem;
Chataigner et al.~\cite{CSW07} proposed a $2$-approximation algorithm for vertex-weighted {\sc max-min $3$-BGP} on $3$-connected graphs,
and a $2$-approximation algorithm for vertex-weighted {\sc max-min $4$-BGP} on $4$-connected graphs.
Approximation algorithms for the vertex-weighted {\sc $k$-BGP} problem on some {\em special classes} of graphs can be found in \cite{Wu11,Wu13,WZW13}.
Recently, on general vertex weighted graphs, Chen et al.~\cite{CCC19} showed that
the algorithm by Chleb{\'{i}}kov{\'{a}}~\cite{Chl96} is also a $5/4$-approximation algorithm for the vertex-weighted {\sc $2$-BGP} problem;
and they presented a $3/2$-approximation algorithm for the vertex-weighted {\sc $3$-BGP} problem
and a $5/3$-approximation algorithm for the vertex-weighted {\sc max-min $3$-BGP} problem.

Motivated by an expensive computation performed by the computational topology software RIVET~\cite{LW15}, 
Madkour et al.~\cite{MNW17} introduced the edge-weighted variant of the {\sc $k$-BGP} problem, denoted as {\sc $k$-eBGP}.
Given an edge non-negatively weighted connected graph $G = (V, E)$,
the weight of a tree subgraph $T$ of $G$ is measured as the total weight of the edges in $T$,
and the weight of a spanning $k$-forest $\{T_1, T_2, \ldots, T_k\}$ is defined as the maximum weight among the $k$ trees.
The {\sc $k$-eBGP} problem is to find a {\em minimum} weight spanning $k$-forest of $G$,
and it can be re-stated as asking for a partition of the vertex set $V$ into $k$ non-empty parts $V_1, V_2, \ldots, V_k$
such that for each part $V_i$ the induced subgraph $G[V_i]$ is connected and its weight is measured as the weight of the minimum spanning tree of $G[V_i]$,
with the objective to minimize the maximum weight of the $k$ parts.
Madkour et al.~\cite{MNW17} showed that the {\sc $k$-eBGP} problem is NP-hard on general graphs for any fixed $k \ge 2$,
and proposed two $k$-approximation algorithms.
Vaishali et al.~\cite{VAP18} presented an $O(k |V|^3)$-time exact algorithm when the input graph is a tree, 
and proved that the problem remains NP-hard on edge uniformly weighted (or unweighted) graphs.
It follows that our {\sc $k$-BGP} problem is NP-hard (again), for any fixed $k \ge 2$.
However, the two $k$-approximation algorithms for {\sc $k$-eBGP} do not trivially work for our {\sc $k$-BGP} problem.

There are works more distantly related to ours.
Andersson et al.~\cite{AGL03} considered the special case of the {\sc $k$-eBGP} problem that arises from applications in shipbuilding industry,
where the vertices are points in the two-dimensional plane and the weight of an edge is the Euclidean distance between the two endpoints.
They showed that this special case remains NP-hard for any constant $k \ge 2$ and presented an $O(|V| \log |V|)$-time approximation algorithm,
which has a worst-case performance ratio of $\frac 43 + \epsilon$ when $k = 2$, and a ratio of $2 + \epsilon$ when $k \ge 3$, for any $\epsilon > 0$.

In a slightly more general case where the input graph is an edge-weighted complete graph and the non-negative edge weights satisfy triangle inequalities,
Guttmann-Beck and Hassin~\cite{GH97} considered a constrained version of the {\sc $k$-eBGP} problem
in which the vertex set $V$ must be partitioned into $k$ equal-sized parts.
They proved that this constrained variant
(as well as another objective~\cite{GH98} to minimize the total weight of the $k$ trees inside the spanning $k$-forest)
is NP-hard even for $k = 2$ and presented an $O(|V|^3)$-time $(k + \epsilon)$-approximation algorithm, for any $\epsilon > 0$.
Motivated by applications from wireless sensor networks, cooperative robotics and music information retrieval,
Caraballo et al.~\cite{CDK18} investigated an alternative quality measure of a part of the vertex set partition,
which is the ratio between the minimum edge weight of its outgoing edges and the maximum edge weight of its minimum spanning tree;
they proposed an $O(k^2 |V|^3)$-time exact algorithm for this variant.

An et al.~\cite{AFK17} studied a tree partition problem to remove at most a given $b$ edges from the input tree,
so that the resulting components can be grouped into $k$ groups of desired orders.
They showed that the problem is NP-complete even if these $k$ groups have the same order of $|V|/k$.
Some other graph partition problems that are more distantly related to our {\sc $k$-BGP} problem have been examined by Cordone and Maffioli~\cite{CM04}.
Kanj et al.~\cite{KKS18} studied a class of graph bi-partition problems ({\it i.e.}, $k = 2$) from fixed-parameter algorithms perspective.

This paper focuses on designing approximation algorithms for the vertex uniformly weighted (or unweighted) {\sc $k$-BGP} problem for a fixed $k \ge 4$,
{\it i.e.}, to minimize the maximum cardinality of the $k$ parts in a partition.
One can probably easily see a trivial $k$-approximation algorithm, since the maximum cardinality is always at least one $k$-th of the order of the input graph.
We remark that the $3/2$-approximation algorithm for the vertex-weighted {\sc $3$-BGP} problem by Chen et al.~\cite{CCC19}
could not be extended trivially for {\sc $k$-BGP} for $k \ge 4$.
After some preliminaries introduced in Section 2,
we present in Section 3 another $3/2$-approximation algorithm for {\sc $3$-BGP} based on two intuitive local improvement operations,
and extend it to become a $k/2$-approximation algorithm for {\sc $k$-BGP}, for any fixed $k \ge 4$.
In Section 4, we introduce several complex local improvement operations for {\sc $4$-BGP},
and use them to design a $24/13$-approximation algorithm.
We conclude the paper in Section 5.

\section{Preliminaries}
Recall that the {\sc $k$-BGP} problem seeks for a partition of the vertex set $V$ of the given connected graph $G = (V, E)$
into $k$ non-empty subsets $V_1, V_2, \ldots, V_k$ such that $G[V_i]$ is connected for every $i = 1, 2, \ldots, k$, and $\max_{1 \le i \le k} |V_i|$ is minimized.
For convenience, we call $\max_{1 \le i \le k} |V_i|$ the {\em size} of the partition $\{V_1, V_2, \ldots, V_k\}$.
In the rest of the paper, when we know these cardinalities, we always assume they are sorted into $0 < |V_1| \le |V_2| \le \ldots \le |V_k|$,
and thus the size of the partition is $|V_k|$.

For two partitions $\{V_1, V_2, \ldots, V_k\}$ and $\{V'_1, V'_2, \ldots, V'_k\}$,
if their sizes $|V'_k| < |V_k|$, or if $|V'_k| = |V_k|$ and $|V'_{k-1}| < |V_{k-1}|$,
then we say the partition $\{V'_1, V'_2, \ldots, V'_k\}$ is {\em better} than the partition $\{V_1, V_2, \ldots, V_k\}$.

For any two disjoint subsets $V_1, V_2 \subset V$, $E(V_1, V_2) \subseteq E$ denotes the edge subset between $V_1$ and $V_2$;
if $E(V_1, V_2) \ne \emptyset$, then we say $V_1$ and $V_2$ are {\em adjacent}.
If additionally both $G[V_1]$ and $G[V_2]$ are connected, then we also say $G[V_1]$ and $G[V_2]$ are {\em adjacent}.%
\footnote{Basically, we reserve the word ``connected'' for a graph and the word ``adjacent'' for two objects with at least one edge between them.}

We note that obtaining an initial feasible partition of $V$ is trivial in $O(|V| + |E|)$ time, as follows:
one first constructs a spanning tree $T$ of $G$,
then arbitrarily removes $k-1$ edges from $T$ to produce a forest of $k$ trees $T_1, T_2, \ldots, T_k$,
and lastly sets $V_i$ to be the vertex set of $T_i$.
The following approximation algorithms all start with a feasible partition and iteratively apply some local improvement operations to improve it.
For $k = 3$, there are only two intuitive local improvement operations and the performance analysis is relatively simple;
for $k = 4$, we introduce several more local improvement operations and the performance analysis is more involved,
though the key ideas in the design and analysis remain intuitive.

Given a connected graph $G = (V, E)$, let $n = |V|$ denote its order.
Let $\OPT$ denote the size of an optimal $k$-part partition of the vertex set $V$.
The following lower bound on $\OPT$ is trivial,
and thus the {\sc $k$-BGP} problem admits a trivial $k$-approximation.

\begin{lemma}
\label{lemma01}
Given a connected graph $G = (V, E)$, $\OPT \ge \frac 1k n$.
\end{lemma}

\section{A $k/2$-approximation for {\sc $k$-BGP}, for a fixed $k \ge 3$}
We consider first $k = 3$, and let $\{V_1, V_2, V_3\}$ denote an initial feasible tripartition (with $|V_1| \le |V_2| \le |V_3|$).
Our goal is to reduce the cardinality of $V_3$ to be no larger than $\frac 12 n$.
It will then follow from Lemma~\ref{lemma01} that the achieved tripartition is within $\frac 32$ of the optimum.

Recently, Chen et al.~\cite{CCC19} presented a $3/2$-approximation algorithm for the vertex-weighted {\sc $3$-BGP} problem,
by noticing that a feasible tripartition ``cuts'' into at most two blocks (that is, maximal $2$-connected components) in the input graph.
It is surely a $3/2$-approximation algorithm for our vertex unweighted {\sc $3$-BGP} problem too,
but no better analysis can be achieved since the algorithm (re-)assigns weights to the cut vertices.
Furthermore, it is noted by the authors that the algorithm cannot be extended trivially for {\sc $k$-BGP} for $k \ge 4$,
for which one has to deal with vertex-weighted graphs having exactly three blocks.

Our new $3/2$-approximation algorithm for {\sc $3$-BGP}, denoted as {\sc Approx-$3$} and detailed in the following, does not deal with blocks,
and it can be extended to become a $k/2$-approximation for {\sc $k$-BGP} for any fixed $k \ge 4$.

Clearly, during the execution of the algorithm {\sc Approx-$3$},
if $|V_3| \le \frac 12 n$, then we may terminate and return the achieved tripartition;
otherwise, we will execute one of the two local improvement operations called {\em Merge} and {\em Pull}, defined in the following, whenever applicable.

Since the input graph $G$ is connected, for any feasible tripartition $\{V_1, V_2, V_3\}$, $V_3$ is adjacent to at least one of $V_1$ and $V_2$.

\begin{definition}
\label{def01}
Operation {\em Merge($V_1, V_2$):}
\begin{itemize}
\parskip=0pt
\item
	{\em precondition:}  $|V_3| > \frac 12 n$;
	$V_1$ and $V_2$ are adjacent;
\item
	{\em effect:} the operation produces a new tripartition $\{V_1 \cup V_2, V_{31}, V_{32}\}$,
	where $\{V_{31}, V_{32}\}$ is an arbitrary feasible bipartition of $V_3$.
\end{itemize}
\end{definition}

\begin{lemma}
\label{lemma02}
Given a connected graph $G = (V, E)$ and a tripartition $\{V_1, V_2, V_3\}$ of the vertex set $V$ with $|V_3| > \frac 12 n$,
the achieved partition by the operation {\em Merge($V_1, V_2$)} is feasible and better.
\end{lemma}
\begin{proof}
Note from the precondition of the operation Merge($V_1, V_2$) that the size of the new part $V_1 \cup V_2$ is $|V_1| + |V_2| < \frac 12 n < |V_3|$;
the sizes of the other two new parts $V_{31}$ and $V_{32}$ partitioned from $V_3$ are clearly strictly less than $|V_3|$.
This proves the lemma.
\end{proof}

\begin{definition}
\label{def02}
Operation {\em Pull($U \subset V_3, V_i$)}, where $i \in \{1, 2\}$,
\begin{itemize}
\parskip=0pt
\item
	{\em precondition:}  $|V_3| > \frac 12 n$;
	both $G[U]$ and $G[V_3 \setminus U]$ are connected, $U$ is adjacent to $V_i$, and $|V_i| + |U| < |V_3|$;
\item
	{\em effect:} the operation produces a new tripartition $\{V_3 \setminus U, V_i \cup U, V_{3-i}\}$ (see for an illustration in Figure~\ref{fig01}).
\end{itemize}
\end{definition}

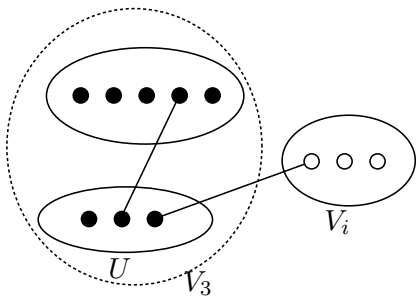
\begin{figure}[htpb]
\begin{center}
  \setlength{\unitlength}{0.68bp}%
  \begin{picture}(271.95, 165.60)(0,0)
  \put(0,0){\includegraphics[scale=0.7]{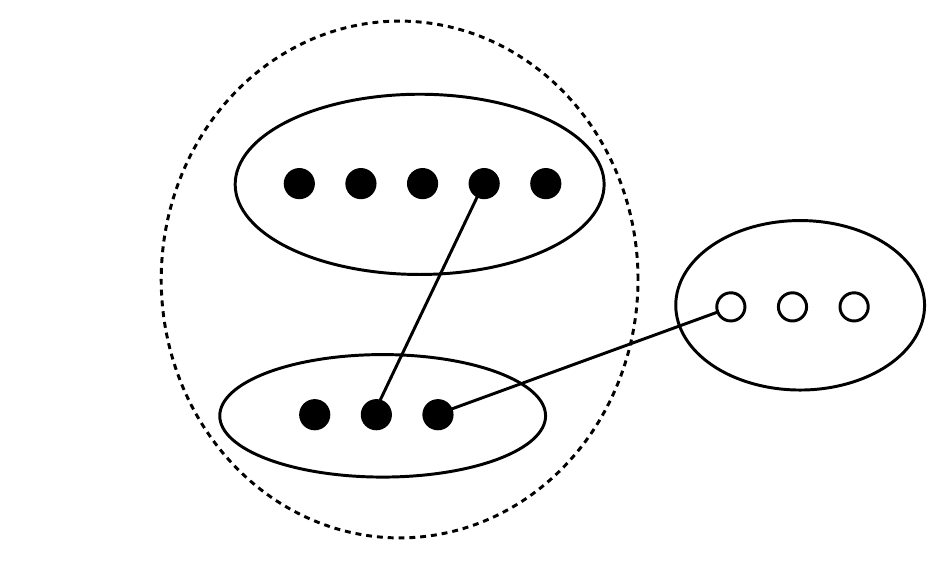}}
  \put(104.02,15.44){\fontsize{11.38}{13.66}\selectfont $U$}
  \put(145.64,8.12){\fontsize{11.38}{13.66}\selectfont $V_3$}
  \put(223.81,42.06){\fontsize{11.38}{13.66}\selectfont $V_i$}
  \end{picture}%
\end{center}
\caption{An illustration of the operation Pull($U \subset V_3, V_i$) that transforms the tripartition $\{V_1, V_2, V_3\}$ to a better tripartition
	$\{V_3 \setminus U, V_i \cup U, V_{3-i}\}$.\label{fig01}}
\end{figure}

\begin{lemma}
\label{lemma03}
Given a connected graph $G = (V, E)$ and a tripartition $\{V_1, V_2, V_3\}$ of the vertex set $V$ with $|V_3| > \frac 12 n$,
the achieved partition by the operation {\em Pull($U \subset V_3, V_i$)} is feasible and better.
\end{lemma}
\begin{proof}
From the precondition of the operation Pull($U \subset V_3, V_i$)
we conclude that the achieved new tripartition $\{V_3 \setminus U, V_i \cup U, V_{3-i}\}$ is feasible;
also, since $|V_i| + |U| < |V_3|$ and $|V_{3-i}| < \frac 12 n < |V_3|$, its size is strictly less than $|V_3|$.
This proves the lemma.
\end{proof}

\begin{lemma}
\label{lemma04}
Given a connected graph $G = (V, E)$,
when none of the {\em Merge} and {\em Pull} operations is applicable to the tripartition $\{V_1, V_2, V_3\}$ of the vertex set $V$ with $|V_3| > \frac 12 n$,
\begin{itemize}
\parskip=0pt
\item[1)]
	$|V_1| + |V_2| < \frac 12 n$ (and thus $|V_1| < \frac 14 n$);
	$V_1$ and $V_2$ aren't adjacent (and thus both are adjacent to $V_3$);
\item[2)]
	let $(u, v) \in E(V_3, V_1)$;
	then $G[V_3 \setminus \{u\}]$ is disconnected;
	suppose $G[V^u_{31}], G[V^u_{32}]$, $\ldots, G[V^u_{3\ell}]$ are the components in $G[V_3 \setminus \{u\}]$,
	then for every $i$, $|V^u_{3i}| \le |V_1|$, and $V^u_{3i}$ and $V_1$ aren't adjacent;
\item[3)]
	no vertex of $V_1 \cup V_2$ is adjacent to any vertex of $V_3$ other than $u$.
\end{itemize}
\end{lemma}
\begin{proof}
See for an illustration in Figure~\ref{fig02}.

\begin{figure}[htpb]
\begin{center}
  \setlength{\unitlength}{0.7bp}%
  \begin{picture}(284.85, 160.51)(0,0)
  \put(0,0){\includegraphics[scale=0.7]{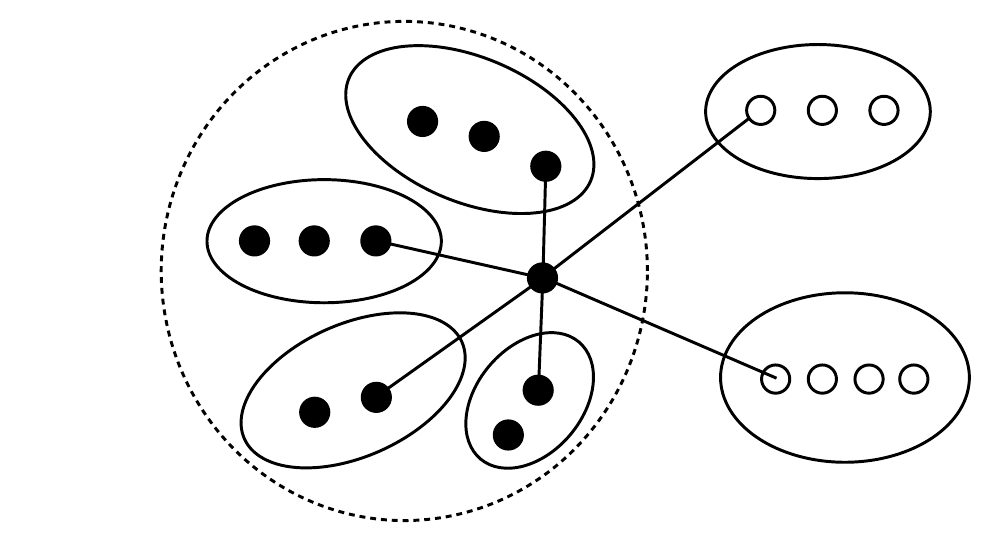}}
  \put(164.21,79.94){\fontsize{11.38}{13.66}\selectfont $u$}
  \put(145.64,8.12){\fontsize{11.38}{13.66}\selectfont $V_3$}
  \put(232.41,93.66){\fontsize{11.38}{13.66}\selectfont $V_1$}
  \put(236.71,16.26){\fontsize{11.38}{13.66}\selectfont $V_2$}
  \end{picture}%
\end{center}
\caption{An illustration of the connectivity configuration of the graph $G = (V, E)$,
	with respect to the tripartition $\{V_1, V_2, V_3\}$ and $|V_3| > \frac 12 n$,
	on which no Merge or Pull operation is applicable.\label{fig02}}
\end{figure}

From $|V_3| > \frac 12 n$, we know $|V_1| + |V_2| < \frac 12 n$ and thus $|V_1| < \frac 14 n$.
Since no Merge operation is possible, $V_1$ and $V_2$ aren't adjacent and consequently they both are adjacent to $V_3$.
This proves Item 1).

Item 2) can be proven similarly as Lemma~\ref{lemma03}.
If $G[V_3 \setminus \{u\}]$ were connected, then it would enable the operation Pull($\{u\} \subset V_3, V_1$), assuming non-trivially $n \ge 5$;
secondly, if $|V^u_{3i}| > |V_1|$ for some $i$, then it would enable the operation Pull($V_3 \setminus V^u_{3i} \subset V_3, V_1$),
since $|V_3 \setminus V^u_{3i}| + |V_1| < |V_3|$;
lastly, if $V^u_{3i}$ and $V_1$ were adjacent for some $i$, then it would enable the operation Pull($V^u_{3i} \subset V_3, V_1$),
since $|V^u_{3i}| + |V_1| \le 2 |V_1| < \frac 12 n < |V_3|$.
This proves the item.

For Item 3), the above item 2) says that $u$ is the only vertex to which a vertex of $V_1$ can possibly be adjacent.
Recall that $V_2$ and $V_3$ are adjacent;
we want to prove that for every $i$, $V^u_{3i}$ and $V_2$ aren't adjacent.
Assume $V_2$ is adjacent to $V^u_{3i}$ for some $i$.
Then, due to $|V^u_{3i}| \le |V_1|$, we have $|V^u_{3i}| + |V_2| \le |V_1| + |V_2| < \frac 12 n < |V_3|$,
suggesting an operation Pull($V^u_{3i} \subset V_3, V_2$) is applicable, a contradiction.
That is, $u$ is the only vertex to which a vertex of $V_2$ can possibly be adjacent.

This finishes the proof.
\end{proof}

From Lemmas~\ref{lemma02}--\ref{lemma04},
we can design an algorithm, denoted as {\sc Approx-$3$},
to first compute in $O(|V| + |E|)$ time an initial feasible tripartition of the vertex set $V$ to the {\sc $3$-BGP} problem;
we then apply the operations Merge and Pull to iteratively reduce the size of the tripartition,
until either this size is no larger than $\frac 12 n$ or none of the two operations is applicable.
The final achieved tripartition is returned as the solution.
See Figure~\ref{fig03} for a high-level description of the algorithm {\sc Approx-$3$}.
We thus conclude with Theorem~\ref{thm01}.

\begin{figure}[ht]
\begin{center}
\framebox{
\begin{minipage}{5.2in}
	The algorithm {\sc Approx-$3$} for {\sc $3$-BGP} on graph $G = (V, E)$:
	\begin{description}
	\parskip=0pt
	\item[Step 1.]
		Construct the initial feasible tripartition $\{V_1, V_2, V_3\}$ of $V$;
	\item[Step 2.]
		while $|V_3| > \frac 12 n$, using Lemma~\ref{lemma04},\\
		\hspace{1in} if a Merge or a Pull operation is applicable, then update the tripartition;
	\item[Step 3.]
		return the final tripartition $\{V_1, V_2, V_3\}$.
	\end{description}
\end{minipage}}
\end{center}
\caption{A high-level description of the algorithm {\sc Approx-$3$} for {\sc $3$-BGP}.\label{fig03}}
\end{figure}

\begin{theorem}
\label{thm01}
The algorithm {\sc Approx-$3$} is an $O(|V||E|)$-time $\frac 32$-approximation for the {\sc $3$-BGP} problem,
and the ratio $\frac 32$ is tight for the algorithm.
\end{theorem}
\begin{proof}
Note that in order to apply a Pull operation using Lemma~\ref{lemma04},
one can execute a graph traversal on $G[V_3 \setminus \{u\}]$ to determine whether it is connected,
and if not, to explore all its connected components.
Such a graph traversal can be done in $O(|V| + |E|)$ time.
A merge operation is also done in $O(|V| + |E|)$ time.
The total number of Merge and Pull operations executed in the algorithm is in $O(|V|)$.
Therefore, the total running time of the algorithm {\sc Approx-$3$} is in $O(|V| |E|)$.

At termination, if $|V_3| \le \frac 12 n$, then by Lemma~\ref{lemma01} we have $\frac {|V_3|}{\OPT} \le \frac 32$.

If $|V_3| > \frac 12 n$, then $|V_1| + |V_2| < \frac 12 n$ and thus $|V_1| < \frac 14 n$, suggesting by Lemma~\ref{lemma02} that $V_1$ and $V_2$ aren't adjacent.
Therefore, both $V_1$ and $V_2$ are adjacent to $V_3$.
By Lemma~\ref{lemma04}, let $u$ denote the unique vertex of $V_3$ to which the vertices of $V_1 \cup V_2$ can be adjacent.
We conclude from Lemma~\ref{lemma04} that $G[V_3 \setminus \{u\}]$ is disconnected,
there are at least two components in $G[V_3 \setminus \{u\}]$ denoted as $G[V^u_{31}], G[V^u_{32}], \ldots, G[V^u_{3\ell}]$ ($\ell \ge 2$),
such that for each $i$, $V^u_{3i}$ is not adjacent to $V_1$ or $V_2$ and $|V^u_{3i}| \le |V_1|$.
That is, $G = (V, E)$ has a very special ``star''-like structure,
in that these $\ell+2$ vertex subsets $V_1, V_2, V^u_{31}, V^u_{32}, \ldots, V^u_{3\ell}$ are pairwise non-adjacent to each other,
but they all are adjacent to the vertex $u$.
Clearly, in an optimal tripartition, the part containing the vertex $u$ has its size at least $|V_3|$,
suggesting the optimality of the achieved partition $\{V_1, V_2, V_3\}$.

For the tightness, one can consider a simple path of order $12$: $v_1$-$v_2$-$v_3$-$\cdots$-$v_{11}$-$v_{12}$,
on which the algorithm {\sc Approx-$3$} may terminate at a tripartition of size $6$,
while an optimal tripartition has size $4$.
This proves the theorem.
\end{proof}

\begin{theorem}
\label{thm02}
The {\sc $k$-BGP} problem admits an $O(|V| |E|)$-time $\frac k2$-approximation, for any constant $k \ge 3$.
\end{theorem}
\begin{proof}
Notice that we may apply the algorithm {\sc Approx-$3$} on the input graph $G = (V, E)$ to obtain a tripartition $\{V_1, V_2, V_3\}$ of the vertex set $V$,
with $|V_1| \le |V_2| \le |V_3|$.

If $|V_3| \le \frac 12 n$, then we may continue on to further partition the largest existing part into two smaller parts iteratively,
resulting in a $k$-part partition in which the size of the largest part is no larger than $\frac 12 n$ (less than $\frac 12 n$ when $k \ge 4$).

If $|V_3| > \frac 12 n$, then let $u$ be the only vertex of $V_3$ to which the vertices of $V_1 \cup V_2$ can be adjacent;
that is, $G[V \setminus \{u\}]$ is disconnected,
there are $\ell \ge 4$ connected components in $G[V \setminus \{u\}]$ (see Figure~\ref{fig02}),
each is adjacent to $u$ and the largest (which is $G[V_2]$) has size less than $\frac 12 n$ (all the others have sizes less than $\frac 14 n$).
When $k \le \ell$, we can achieve a $k$-part partition by setting the $k-1$ largest components to be the $k-1$ parts,
and all the other components together with $u$ to be the last part.
Such a partition has size no greater than $\max\{|V_2|, \OPT\}$,
since in an optimal $k$-part partition the part containing the vertex $u$ is no smaller than the last constructed part.
When $k > \ell$, we can start with the $\ell$-part partition obtained as above to further partition the largest existing part into two smaller parts iteratively,
resulting in a $k$-part partition in which the size of the largest part is less than $\frac 12 n$.

In summary, we either achieve an optimal $k$-part partition or achieve a $k$-part partition in which the size of the largest part is no greater than $\frac 12 n$.
Using the lower bound in Lemma~\ref{lemma01}, this is a $\frac k2$-approximation.

Running {\sc Approx-$3$} takes $O(|V| |E|)$ time;
the subsequent iterative bipartitioning needs only $O(|E|)$ per iteration.
Therefore, the total running time is still in $O(|V| |E|)$, since $k \le |V|$.

Lastly, we remark that in the above proof, when $k \ge 4$, if the achieved $k$-part partition is not optimal, then its size is less than $\frac 12 n$.
That is, when $k \ge 4$, the ratio $\frac k2$ is not tight for the approximation algorithm.
\end{proof}

\section{A $24/13$-approximation for {\sc $4$-BGP}}
Theorem~\ref{thm02} states that the {\sc $4$-BGP} problem admits a $2$-approximation.
In this section, we design a better $\frac {24}{13}$-approximation,
which uses three more local improvement operations besides the similarly defined Merge and Pull operations.
Basically, these three new operations each finds a subset of the largest two parts, respectively, to merge them into a new part.

Let $\{V_1, V_2, V_3, V_4\}$ denote an initial feasible tetrapartition, with $|V_1| \le |V_2| \le |V_3| \le |V_4|$.
Note that these four parts must satisfy some adjacency constraints due to $G$ being connected.
We try to reduce the size of $V_4$ to be no larger than $\frac 25 n$, whenever possible;
or otherwise we will show that the achieved partition is a $\frac {24}{13}$-approximation.
We point out a major difference from {\sc $3$-BGP},
that the two largest parts $V_3$ and $V_4$ in a tetrapartition can both be larger than the desired bound of $\frac 25 n$.
Therefore, we need new local improvement operations.

In the following algorithm denoted as {\sc Approx-$4$},
if $|V_4| \le \frac 25 n$, then we may terminate and return the achieved tetrapartition;
it follows from Lemma~\ref{lemma01} that the achieved tetrapartition is within $\frac 85$ of the optimum.
Otherwise, the algorithm will execute one of the following local improvement operations whenever applicable.

The first local improvement operation is similar to the Merge operation designed for {\sc Approx-$3$},
except that it now deals with more cases.

\begin{definition}
\label{def03}
Operation {\em Merge($V_i, V_j$)}, for some $i, j \in \{1, 2, 3\}$:
\begin{itemize}
\parskip=0pt
\item
	{\em precondition:}  $|V_4| > \frac 25 n$;
	$V_i$ and $V_j$ are adjacent, and $|V_i| + |V_j| < |V_4|$;
\item
	{\em effect:} the operation produces a new tetrapartition $\{V_i \cup V_j, V_{6-i-j}, V_{41}, V_{42}\}$,
	where $\{V_{41}, V_{42}\}$ is an arbitrary feasible bipartition of $V_4$.
\end{itemize}
\end{definition}

\begin{lemma}
\label{lemma05}
Given a connected graph $G = (V, E)$ and a tetrapartition $\{V_1, V_2, V_3, V_4\}$ of the vertex set $V$ with $|V_4| > \frac 25 n$,
the achieved tetrapartition by the operation {\em Merge($V_i, V_j$)}, for some $i, j \in \{1, 2, 3\}$, is feasible and better.
\end{lemma}
\begin{proof}
The proof is similar to the proof of Lemma~\ref{lemma02}.

Note from the precondition of the operation Merge($V_i, V_j$) that the size of the new part $V_i \cup V_j$ is $|V_i| + |V_j| < |V_4|$;
the sizes of the other two new parts $V_{41}$ and $V_{42}$ partitioned from $V_4$ are clearly strictly less than $|V_4|$.
Let $h = 6 - i - j$;
it follows that if $|V_h| = |V_4|$ (which implies $h = 3$),
then the size of the largest part is unchanged but the size of the second largest part reduces by at least $1$;
if $|V_h| < |V_4|$, then the size of the largest part reduces by at least $1$.
Therefore, the new partition is better.
This proves the lemma.
\end{proof}

The next local improvement operation is very similar to the Pull operation designed for {\sc Approx-$3$},
except that it now deals with more cases.
See for an illustration in Figure~\ref{fig01}, with $V_3$ replaced by $V_j$.

\begin{definition}
\label{def04}
Operation {\em Pull($U \subset V_j, V_i$)}, for some pair $(i, j) \in \{(1, 3), (1, 4)$, $(2, 3), (2, 4), (3, 4)\}$,
\begin{itemize}
\parskip=0pt
\item
	{\em precondition:}  $|V_4| > \frac 25 n$ and no Merge operation is applicable;
	both $G[U]$ and $G[V_j \setminus U]$ are connected, $U$ is adjacent to $V_i$, and $|V_i| + |U| < |V_j|$;
\item
	{\em effect:} the operation produces a new tetrapartition $\{V_j \setminus U, V_i \cup U, V_a, V_b\}$, where $a, b \in \{1, 2, 3, 4\} \setminus \{i, j\}$.
\end{itemize}
\end{definition}

\begin{lemma}
\label{lemma06}
Given a connected graph $G = (V, E)$ and a tetrapartition $\{V_1, V_2, V_3, V_4\}$ of the vertex set $V$ with $|V_4| > \frac 25 n$,
the achieved tetrapartition by the operation {\em Pull($U \subset V_j, V_i$)} is feasible and better.
\end{lemma}
\begin{proof}
Note from the precondition that $|V_4| > \frac 25 n$;
thus $|V_1| + |V_2| + |V_3| < \frac 35 n$, and further $|V_1| + |V_2| < \frac 25 n$.
Since no Merge operation is applicable at the time Pull($U \subset V_j, V_i$) is performed, $V_1$ and $V_2$ aren't adjacent.
The new partition is feasible since $G[V_j \setminus U]$ and $G[V_i \cup U]$ are connected.

The sizes of the two new parts $V_j \setminus U$ and $V_i \cup U$ are less than $|V_j|$;
the sizes of the other two parts $V_a$ and $V_b$ are unchanged.
If $|V_3| = |V_4|$ (implying $(i, j) \ne (3, 4)$, and thus $i \in \{1, 2\}$),
	then the size of the largest part is unchanged but due to $|V_2| < |V_3|$ the size of the second largest part reduces by at least $1$;
if $|V_3| < |V_4|$ and $j = 4$, then the size of the largest part reduces by at least $1$;
if $|V_3| < |V_4|$ and $j = 3$, then $|V_2| < |V_3|$ and the size of the largest part is unchanged but the size of the second largest part reduces by at least $1$.
Therefore, the new partition is better.
This proves the lemma.
\end{proof}

\begin{lemma}
\label{lemma07} {\em (Structure Properties)}
Given a connected graph $G = (V, E)$,
when none of the {\em Merge} and {\em Pull} operations is applicable to
the tetrapartition $\{V_1, V_2, V_3, V_4\}$ of the vertex set $V$ with $|V_4| > \frac 25 n$,
\begin{itemize}
\item[1)]
	$|V_1| < \frac 15 n$, $|V_2| < \frac 3{10} n$, $|V_1| + |V_2| < \frac 25 n$, and $V_1$ and $V_2$ aren't adjacent;
\item[2)]
	if $V_i$ and $V_3$ are adjacent, for some $i \in \{1, 2\}$, then $|V_i| + |V_3| \ge |V_4|$;
\item[3)]
	if $V_i$ and $V_4$ are adjacent for some $i \in \{1, 2\}$, and there is an edge $(u, v) \in E(V_4, V_i)$,
	then $G[V_4 \setminus \{u\}]$ is disconnected, every component $G[V^u_{4\ell}]$ in $G[V_4 \setminus \{u\}]$ has its order $|V^u_{4\ell}| \le |V_i|$,
	and $V^u_{4\ell}$ and $V_i$ aren't adjacent;

	furthermore, if $V^u_{4\ell}$ and $V_3$ are adjacent, then $|V^u_{4\ell} \cup V_3| \ge |V_4|$;
\item[4)]
	if $V_i$ and $V_3$ are adjacent for some $i \in \{1, 2\}$, $|V_i| < \frac 13 |V_4|$, and there is an edge $(v, u) \in E(V_3, V_i)$,
	then $G[V_3 \setminus \{v\}]$ is disconnected, every component $G[V^v_{3\ell}]$ in $G[V_3 \setminus \{v\}]$ has its order $|V^v_{3\ell}| \le |V_i|$,
	and $V^v_{3\ell}$ and $V_i$ aren't adjacent;
\item[5)]
	if $|V_2| \ge \frac 16 |V_4|$, then the partition $\{V_1, V_2, V_3, V_4\}$ is a $\frac {24}{13}$-approximation;

	otherwise, we have
\begin{equation}
\label{eq01}
\left\{
\begin{array}{rcl}
|V_2| + |V_3|   &\ge &|V_4|,\\
|V_4| 		    &<   &\frac 12 n,\\
|V_1| \le |V_2| &<   &\frac 16 |V_4| < \frac 1{12} n,\\
|V_3| 			&>   &\frac 13 n;
\end{array}
\right.
\end{equation}
\item[6)]
	if both $V_1$ and $V_2$ are adjacent to $V_j$ for some $j \in \{3, 4\}$, then the vertices of $V_1 \cup V_2$ can be adjacent to only one vertex of $V_j$.
\end{itemize}
\end{lemma}
\begin{proof}
See for an illustration in Figure~\ref{fig04}.

\begin{figure}[htpb]
\begin{center}
  \setlength{\unitlength}{0.7bp}%
  \begin{picture}(373.32, 188.75)(0,0)
  \put(0,0){\includegraphics[scale=0.7]{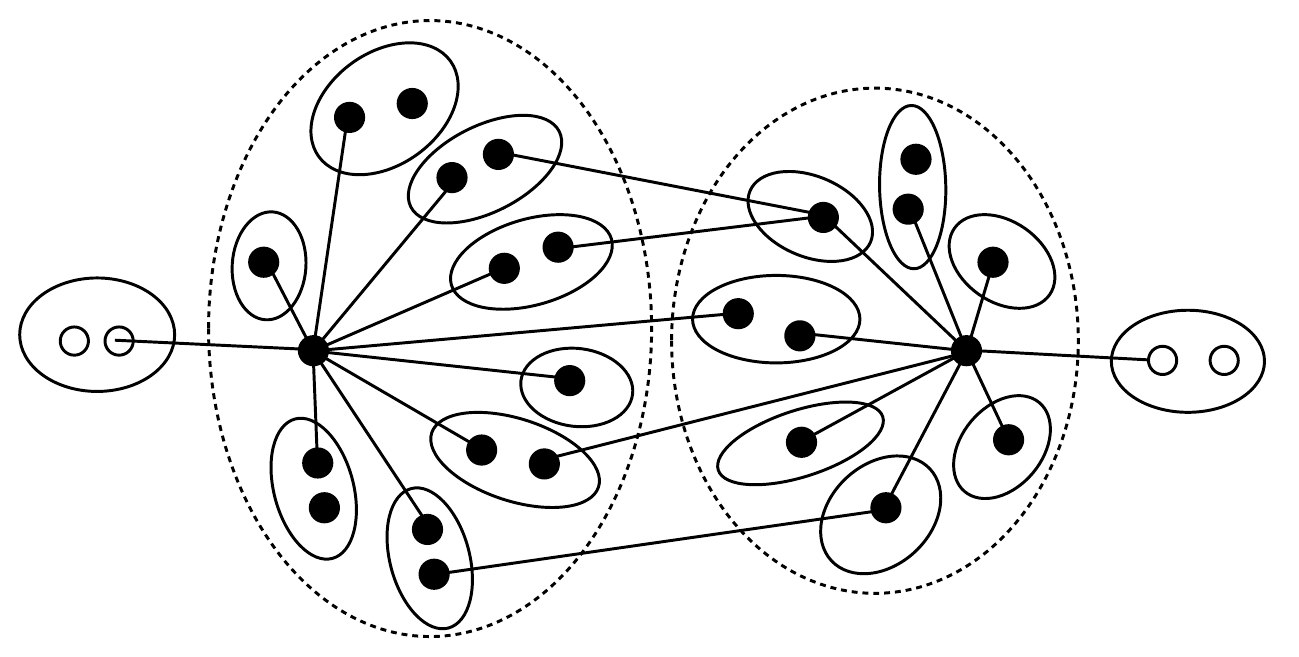}}
  \put(285.33,92.02){\fontsize{11.38}{13.66}\selectfont $v$}
  \put(279.14,15.16){\fontsize{11.38}{13.66}\selectfont $V_3$}
  \put(336.79,58.33){\fontsize{11.38}{13.66}\selectfont $V_1$}
  \put(20.86,62.90){\fontsize{11.38}{13.66}\selectfont $V_2$}
  \put(71.96,9.11){\fontsize{11.38}{13.66}\selectfont $V_4$}
  \put(73.73,78.99){\fontsize{11.38}{13.66}\selectfont $u$}
  \end{picture}%
\end{center}
\caption{An illustration of the connectivity configuration of the graph $G = (V, E)$,
	with respect to the tetrapartition $\{V_1, V_2, V_3, V_4\}$ and $|V_4| > \frac 25 n$,
	on which no Merge or Pull operation is applicable.\label{fig04}}
\end{figure}

Using $|V_1| \le |V_2| \le |V_3| \le |V_4|$ and $|V_4| > \frac 25 n$, we have $|V_1| + |V_2| + |V_3| < \frac 35 n$,
and consequently $|V_1| < \frac 15 n$, $|V_2| < \frac 3{10} n$, $|V_1| + |V_2| < \frac 25 n$.

Items 1) and 2) hold due to no applicable Merge operation (Definition~\ref{def03}).

If $i = 1$, item 3) can be proven similarly as Lemmas~\ref{lemma03} and \ref{lemma04}.
Using $|V_1| < \frac 15 n$ and $|V_4| > \frac 25 n$,
if $G[V_4 \setminus \{u\}]$ is connected, then it would enable the Pull($\{u\} \subset V_4, V_1$) operation (Definition~\ref{def04}, assuming non-trivially $n \ge 5$);
if a component $G[V^u_{4\ell}]$ of $G[V_4 \setminus \{u\}]$ has its order $|V^u_{4\ell}| > |V_1|$,
	then it would enable the Pull($V_4 \setminus V^u_{4\ell} \subset V_4, V_1$) operation;
if $V^u_{4\ell}$ and $V_1$ are adjacent, then it would enable the Pull($V^u_{4\ell} \subset V_4, V_1$) operation.

If $i = 2$ and $G[V_4 \setminus \{u\}]$ is connected,
then it would enable the Pull($\{u\} \subset V_4, V_2$) operation since $|V_2| + 1 < \frac 3{10} n + 1 \le \frac 25 n < |V_4|$ (assuming non-trivially $n \ge 10$).
For a component $G[V^u_{4\ell}]$ of $G[V_4 \setminus \{u\}]$, similarly we have $|V^u_{4\ell}| \le |V_2|$.
If $V^u_{4\ell}$ and $V_2$ are adjacent, then $|V^u_{4\ell}| + |V_2| \ge |V_4|$ since otherwise the Pull($V^u_{4\ell} \subset V_4, V_2$) operation would be applicable.
Also, for the same reason, in this case $V_1$ cannot be adjacent to $V_4$, and thus $V_1$ has to be adjacent to $V_3$.
In summary, we have $2 |V_2| \ge |V^u_{4\ell}| + |V_2| \ge |V_4|$ and $|V_1| + |V_3| \ge |V_4|$, suggesting $|V_4| \le \frac 25 n$, a contradiction.
This contradiction proves that $V^u_{4\ell}$ and $V_2$ aren't adjacent.
A similar contradiction using a Pull operation shows that if $V^u_{4\ell}$ and $V_3$ are adjacent, then $|V^u_{4\ell} \cup V_3| \ge |V_4|$.
The third item is thus proved.

Item 4) can be proven similarly, as follows.
We claim that $|V_3| \ge |V_i| + 2$.
To prove this claim, we see that $|V_3| \le |V_i| + 1$ implies $|V_i| + |V_3| \le 2 |V_i| + 1 < |V_4|$, contradicting item 2).
It follows from the above claim that $G[V_3 \setminus \{v\}]$ is disconnected
since the Pull($\{v\} \subset V_3, V_i$) operation isn't applicable.
For a component $G[V^v_{3\ell}]$ of $G[V_3 \setminus \{v\}]$, similarly we have $|V^v_{3\ell}| \le |V_i|$.
If $|V^v_{3\ell}| + |V_i| \ge |V_3|$, then $2 |V_i| \ge |V_3|$;
from $|V_i| + |V_3| \ge |V_4|$ we have $|V_i| \ge \frac 13 |V_4|$, a contradiction to the presumption that $|V_i| < \frac 13 |V_4|$.
Therefore, $|V^v_{3\ell}| + |V_i| < |V_3|$ holds.
Next, if $V^v_{3\ell}$ and $V_i$ are adjacent, then the Pull($V^v_{3\ell} \subset V_3, V_i$) operation would be applicable;
that is, $V^v_{3\ell}$ and $V_i$ aren't adjacent.
This proves item 4).

To prove item 5), if $V_1$ and $V_2$ are both adjacent to $V_4$, but not to $V_3$,
then by item 2) we know that the vertices of $V_1 \cup V_2$ can be adjacent to only one vertex of $V_4$, say $u$
(otherwise, a Merge operation would be applicable).
Note that $V_3$ must also be adjacent to $V_4$.
If $V_3$ is adjacent to a component $G[V^u_{4\ell}]$ of $G[V_4 \setminus \{u\}]$, then $|V_1| + |V_3| \ge |V^u_{4\ell}| + |V_3| \ge |V_4|$.
Therefore, $|V_2| \ge \frac 16 |V_4|$ implies $|V_4| \le \frac 6{13} n$,
suggesting the partition $\{V_1, V_2, V_3, V_4\}$ is a $\frac {24}{13}$-approximation by Lemma~\ref{lemma01}.
If the vertices of $V_3$ are adjacent to only the vertex $u \in V_4$,
then consider an optimal tetrapartition $\{V^*_1, V^*_2, V^*_3, V^*_4\}$, and assume the vertex $u$ is in $V^*_j$;
clearly, $|V^*_j| \ge |V_4|$, suggesting the partition $\{V_1, V_2, V_3, V_4\}$ is also optimal.

We next discuss the case where at least one of $V_1$ and $V_2$, say $V_i$, is adjacent to $V_3$.
If $i = 1$, or if $|V_1| \ge \frac 16 |V_4|$, then we again have $|V_4| \le \frac 6{13} n$,
suggesting the partition $\{V_1, V_2, V_3, V_4\}$ is a $\frac {24}{13}$-approximation by Lemma~\ref{lemma01}.
In the other case, $|V_1| < \frac 16 |V_4|$ and $V_1$ isn't adjacent to $V_3$ but to $V_4$, and by item 3) suppose $V_1$ is adjacent to the vertex $u \in V_4$.
We further conclude for the same reason that each component $G[V^u_{4\ell}]$ of $G[V_4 \setminus \{u\}]$ cannot be adjacent to either $V_2$ or $V_3$
(otherwise, either a Merge operation would be applicable, or we again have $|V_1| + |V_3| \ge |V^u_{4\ell}| + |V_3| \ge |V_4|$),
and consequently $V_2$ and $V_3$ are adjacent.
Consider next an optimal tetrapartition $\{V^*_1, V^*_2, V^*_3, V^*_4\}$, and assume the vertex $u$ is in $V^*_j$.
If all but one of the components of $G[(V_4 \cup V_1) \setminus \{u\}]$ are in $V^*_j$, then $|V^*_j| \ge |V_4|$,
suggesting the partition $\{V_1, V_2, V_3, V_4\}$ is also optimal;
otherwise, $|V^*_1| \le |V^*_2| \le |V_1|$ and thus from $|V_2| + |V_3| \ge |V_4|$ we have $|V^*_4| \ge \frac 12 (2 - \frac 16) |V_4| = \frac {11}{12} |V_4|$,
suggesting the partition $\{V_1, V_2, V_3, V_4\}$ is a $\frac {12}{11}$-approximation.

In summary, if $|V_2| \ge \frac 16 |V_4|$, then the partition $\{V_1, V_2, V_3, V_4\}$ is a $\frac {24}{13}$-approximation.
Otherwise, we have $|V_2| < \frac 16 |V_4|$.
Furthermore, if one of $V_1$ and $V_2$, say $V_i$, is adjacent to $V_3$, then $|V_2| + |V_3| \ge |V_i| + |V_3| \ge |V_4|$;
if none of $V_1$ and $V_2$ is adjacent to $V_3$, then the same proof earlier shows that $V_3$ is adjacent to a component $G[V^u_{4\ell}]$ of $G[V_4 \setminus \{u\}]$,
suggesting $|V_1| + |V_3| \ge |V^u_{4\ell}| + |V_3| \ge |V_4|$.
It follows that, either way we have $|V_2| + |V_3| \ge |V_4|$ and thus $n = |V_1| + |V_2| + |V_3| + |V_4| > 2 |V_4|$.
This completes the proof of item 5).

The last item 6) can be proven by a simple contradiction by setting a proper $U \subset V_j$ to enable the Pull($U \subset V_j, V_2$) operation.
\end{proof}

\begin{proposition}
\label{prop01}
In the following, we distinguish three cases for the tetrapartition $\{V_1, V_2, V_3, V_4\}$ of the vertex set $V$ with $|V_4| > \frac 25 n$,
to which none of the {\em Merge} and {\em Pull} operations is applicable, and Eq.~(\ref{eq01}) holds:
\begin{description}
\parskip=0pt
\item[{\rm Case 1:}]
	none of $V_1$ and $V_2$ is adjacent to $V_3$
	({\it i.e.}, both $V_1$ and $V_2$ are adjacent to $V_4$ only and at the vertex $u \in V_4$ only;
	see for an illustration in Figs.~\ref{fig05} and \ref{fig06}, to be handled in Theorems~\ref{thm03} and \ref{thm04});
\item[{\rm Case 2:}]
	none of $V_1$ and $V_2$ is adjacent to $V_4$
	({\it i.e.}, both $V_1$ and $V_2$ are adjacent to $V_3$ only and at the vertex $v \in V_3$ only;
	see for an illustration in Figure~\ref{fig07}, to be handled in Theorems~\ref{thm05} and \ref{thm06});
\item[{\rm Case 3:}]
	one of $V_1$ and $V_2$ is adjacent to $V_3$ and the other is adjacent to $V_4$
	(see for an illustration in Figure~\ref{fig08}, to be handled in Theorems~\ref{thm07} and \ref{thm08}).
\end{description}
The final conclusion is presented as Theorem~\ref{thm09}.
\end{proposition}

Lemma~\ref{lemma07} states several structural properties of the graph $G = (V, E)$ with respect to the tetrapartition,
which is yet unknown to be a $\frac {24}{13}$-approximation or not.
For each of the three cases listed in Proposition~\ref{prop01}, Lemma~\ref{lemma07} leads to a further conclusion,
stated separately in Theorems~\ref{thm03}, \ref{thm05}, and \ref{thm07}.

\begin{theorem}
\label{thm03}
In Case 1, let $V'_4$ denote the union of the vertex sets of all the components of $G[V_4 \setminus \{u\}]$ that are adjacent to $V_3$;
if $|V'_4| \le |V_1| + |V_2| + \frac {11}{24} |V_4|$, then the partition $\{V_1, V_2, V_3, V_4\}$ is a $\frac {24}{13}$-approximation.
\end{theorem}
\begin{proof}
From Lemma~\ref{lemma07}, we assume without loss of generality that some component of $G[V_4 \setminus \{u\}]$ is adjacent to $V_3$,
as otherwise the partition $\{V_1, V_2, V_3, V_4\}$ is already optimal.
It follows that $|V_1| + |V_3| \ge |V_4|$.

If $|V'_4| \le |V_1| + |V_2| + \frac {11}{24} |V_4|$ ($< \frac {19}{24} |V_4|$ by Eq.~(\ref{eq01})),
then denote the components of $G[V_4 \setminus \{u\}]$ not adjacent to $V_3$ as $G[V^u_{4i}]$, $i = 1, 2, \ldots \ell$, with $\ell \ge 2$.
In an optimal $4$-partition denoted as $\{V^*_1, V^*_2, V^*_3, V^*_4\}$, assume $u \in V^*_j$.
If $V^*_j$ contains all these $\ell + 2$ subsets, $V_1, V_2, V^u_{4i}, i = 1, 2, \ldots \ell$,
then $|V^*_j| = |V_1| + |V_2| + |V_4| - |V'_4| \ge \frac {13}{24} |V_4|$.
In the other case, at least one of these $\ell + 2$ subsets becomes a separate part in $\{V^*_1, V^*_2, V^*_3, V^*_4\}$,
of which the size is at most $|V_2|$,
and thus we have $|V^*_4| \ge \frac 13 (|V_1| + |V_3| + |V_4|) \ge \frac 23 |V_4|$.
Therefore, we always have $|V^*_4| \ge \frac {13}{24} |V_4|$, and thus the partition $\{V_1, V_2, V_3, V_4\}$ is a $\frac {24}{13}$-approximation.
\end{proof}

\begin{figure}[htpb]
\begin{center}
  \setlength{\unitlength}{0.7bp}%
  \begin{picture}(316.23, 188.75)(0,0)
  \put(0,0){\includegraphics[scale=0.7]{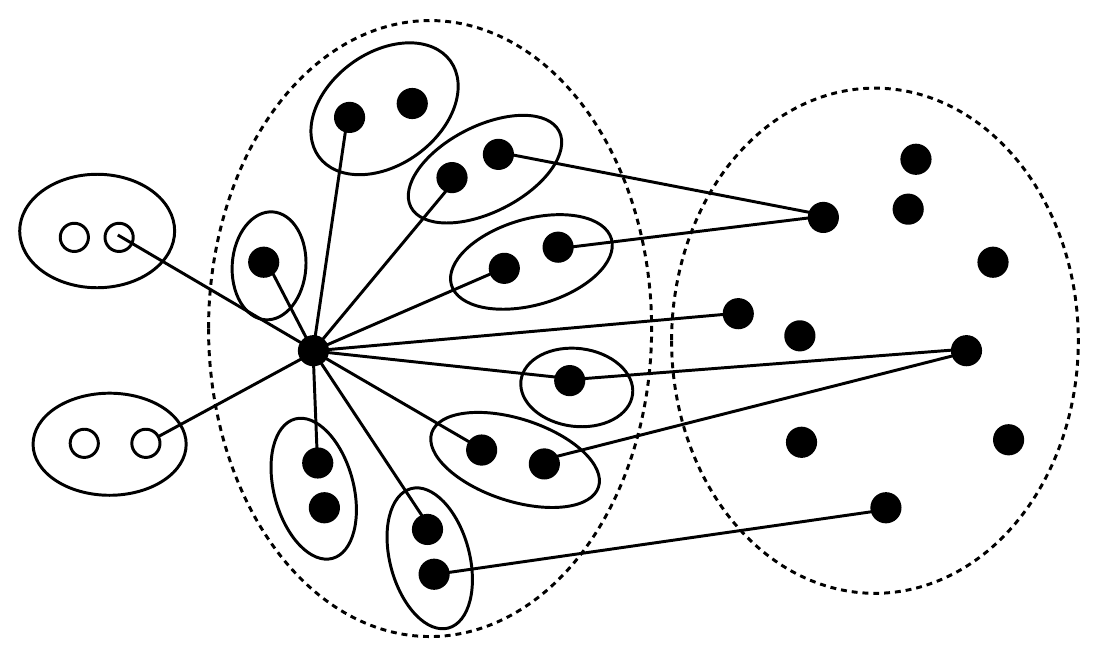}}
  \put(279.14,15.16){\fontsize{11.38}{13.66}\selectfont $V_3$}
  \put(26.23,34.44){\fontsize{11.38}{13.66}\selectfont $V_1$}
  \put(20.86,92.76){\fontsize{11.38}{13.66}\selectfont $V_2$}
  \put(71.96,9.11){\fontsize{11.38}{13.66}\selectfont $V_4$}
  \put(70.44,83.47){\fontsize{11.38}{13.66}\selectfont $u$}
  \end{picture}%
\end{center}
\caption{An illustration of the connectivity configuration of the graph $G = (V, E)$,
	with respect to the tetrapartition $\{V_1, V_2, V_3, V_4\}$ in Case 1.\label{fig05}}
\end{figure}
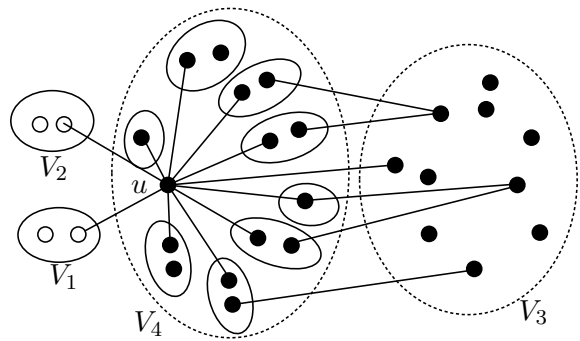

We have seen that $G[V_4]$ exhibits a nice star-like configuration (Figure~\ref{fig05}), due to $V_4$ being adjacent to $V_1$ and $V_2$.
Since none of $V_1$ and $V_2$ is adjacent to $V_3$ in Case 1, the connectivity configuration of $G[V_3]$ is unclear.
We next bipartition $V_3$ as evenly as possible, and let $\{V_{31}, V_{32}\}$ denote the achieved bipartition with $|V_{31}| \le |V_{32}|$.
If $|V_{32}| \le \frac 23 |V_3|$,
and assuming there are multiple components of $G[V_4 \setminus \{u\}]$ adjacent to $V_{3i}$ (for some $i \in \{1, 2\}$) with their total size greater than $|V_1|$,
then we find a minimal sub-collection of these components of $G[V_4 \setminus \{u\}]$ adjacent to $V_{3i}$ with their total size exceeding $|V_1|$,
denote by $V'_4$ the union of their vertex sets,
and subsequently create three new parts $V_4 \cup V_1 \setminus V'_4$, $V'_4 \cup V_{3i}$, and $V_{3,3-i}$,
while keeping $V_2$ unchanged.
One sees that this new tetrapartition is feasible and better, since $|V'_4| + |V_{3i}| \le 2 |V_1| + |V_{3i}| < \frac 13 |V_4| + \frac 23 |V_3| \le |V_4|$.

In the other case, by Lemma~\ref{lemma03}, $G[V_3]$ also exhibits a nice star-like configuration centering at some vertex $v$,
such that $G[V_3 \setminus \{v\}]$ is disconnected and each component of $G[V_3 \setminus \{v\}]$ has size less than $\frac 13 |V_3|$.
See for an illustration in Figure~\ref{fig06}.

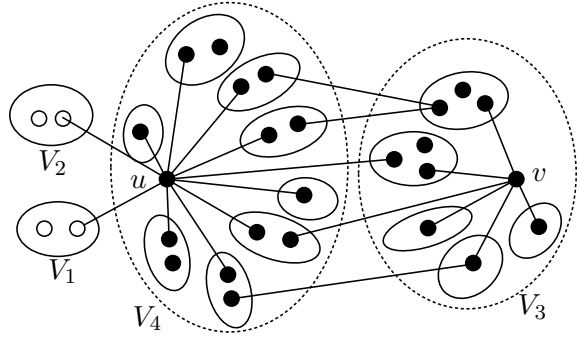
\begin{figure}[htpb]
\begin{center}
  \setlength{\unitlength}{0.7bp}%
  \begin{picture}(316.23, 188.75)(0,0)
  \put(0,0){\includegraphics[scale=0.7]{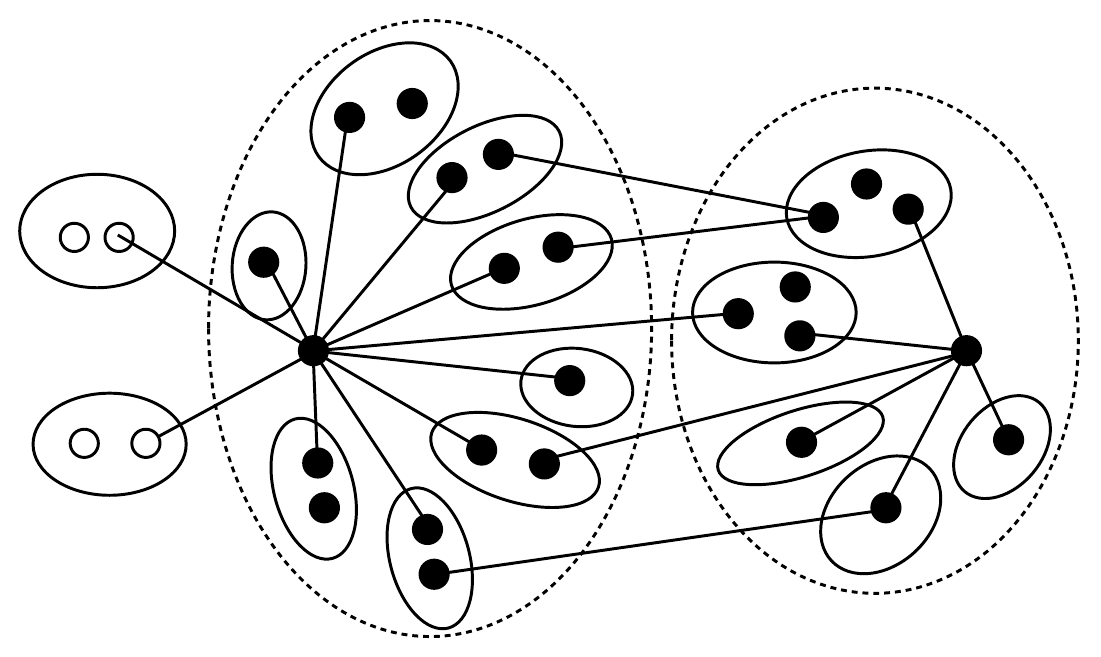}}
  \put(286.70,86.97){\fontsize{11.38}{13.66}\selectfont $v$}
  \put(279.14,15.16){\fontsize{11.38}{13.66}\selectfont $V_3$}
  \put(26.23,34.44){\fontsize{11.38}{13.66}\selectfont $V_1$}
  \put(20.86,92.76){\fontsize{11.38}{13.66}\selectfont $V_2$}
  \put(71.96,9.11){\fontsize{11.38}{13.66}\selectfont $V_4$}
  \put(70.44,83.47){\fontsize{11.38}{13.66}\selectfont $u$}
  \end{picture}%
\end{center}
\caption{An illustration of the ``bi-star''-like configuration of the graph $G = (V, E)$,
	with respect to the tetrapartition $\{V_1, V_2, V_3, V_4\}$ in Case 1.\label{fig06}}
\end{figure}

The following {\em Bridge-1} operation aims to find a subset $V'_3 \subset V_3$ and a subset $V'_4 \subset V_4$ to form a new part larger than $V_1$,
possibly cutting off another subset $V''_4$ from $V_4$ and merging it into $V_3$, and merging the old part $V_1$ into $V_4$.
This way, a better tetrapartition is achieved.
We will prove later that when such a bridging operation isn't applicable,
each component in the residual graph by deleting the two star centers has size at most $2|V_1| + |V_2|$,
and subsequently the tetrapartition can be shown to be a $\frac {12}7$-approximation.

\begin{definition}
\label{def05}
Operation {\em Bridge-1($V_3, V_4$):} 
\begin{itemize}
\parskip=0pt
\item
	{\em precondition:} In Case 1,
	there are multiple components of $G[V_4 \setminus \{u\}]$ adjacent to $V_3$ with their total size greater than $|V_1| + |V_2| + \frac {11}{24} |V_4|$,
	and there is a vertex $v \in V_3$ such that $G[V_3 \setminus \{v\}]$ is disconnected and each component has size less than $\frac 13 |V_3|$.
\item
	{\em effect:} Find a component $G[V^u_{4x}]$ of $G[V_4 \setminus \{u\}]$, if exists,
	that is adjacent to a component $G[V^v_{3y}]$ of $G[V_3 \setminus \{v\}]$;
	initialize $V'_4$ to be $V^u_{4x}$ and $V'_3$ to be $V^v_{3y}$;
	iteratively,
	\begin{itemize}
	\parskip=0pt
	\item
		let ${\cal C}_3$ denote the collection of the components of $G[V_3 \setminus \{v\}]$ that are adjacent to $V'_4$, excluding $V'_3$;
		\begin{itemize}
		\parskip=0pt
		\item
		if the total size of components in ${\cal C}_3$ exceeds $2 |V_1| - |V'_3|$,
		then the operation greedily finds a minimal sub-collection of these components of ${\cal C}_3$ with their total size exceeding $2 |V_1| - |V'_3|$,
		adds their vertex sets to $V'_3$,
		and proceeds to termination;
		\item
		if the total size of components in ${\cal C}_3$ is less than $2 |V_1| - |V'_3|$,
		then the operation adds the vertex sets of all these components to $V'_3$;
	\end{itemize}
	\item
		let ${\cal C}_4$ denote the collection of the components of $G[V_4 \setminus \{u\}]$ that are adjacent to $V'_3$, excluding $V'_4$;
		\begin{itemize}
		\parskip=0pt
		\item
		if the total size of components in ${\cal C}_4$ exceeds $|V_1| - |V'_4|$,
		then the operation greedily finds a minimal sub-collection of these components of ${\cal C}_4$ with their total size exceeding $|V_1| - |V'_4|$,
		adds their vertex sets to $V'_4$,
		and proceeds to termination;
		\item
		if the total size of components in ${\cal C}_4$ is less than $|V_1| - |V'_4|$,
		then the operation adds the vertex sets of all these components to $V'_4$;
		\end{itemize}
	\item
		if both ${\cal C}_3$ and ${\cal C}_4$ are empty, then the operation terminates without updating the partition.
	\end{itemize}
	At termination, exactly one of $|V'_3| > 2 |V_1|$ and $|V'_4| > |V_1|$ holds. 
	\begin{itemize}
	\parskip=0pt
	\item
		When $|V'_3| > 2 |V_1|$, we have $|V'_4| \le |V_1|$ and $|V'_3| < 2 |V_1| + \frac 13 |V_3|$;
		\begin{itemize}
		\parskip=0pt
		\item
			if the collection of the components of $G[V_4 \setminus \{u\}]$ that are adjacent to $V'_3$, excluding $V'_4$, exceeds $|V_1| - |V'_4|$,
			then the operation greedily finds a minimal sub-collection of these components with their total size exceeding $|V_1| - |V'_4|$,
			and denotes by $V''_4$ the union of their vertex sets;
			subsequently, the operation creates three new parts $V_4 \cup V_1 \setminus (V'_4 \cup V''_4)$, $(V'_4 \cup V''_4) \cup V'_3$, and $V_3 \setminus V'_3$;
		\item
			otherwise, the operation greedily finds a minimal sub-collection of the components of $G[V_4 \setminus \{u\}]$ that aren't adjacent to $V'_3$
			with their total size exceeding $|V_1| - |V'_4|$,
			and denotes by $V''_4$ the union of their vertex sets;
			subsequently, the operation creates three new parts $V_4 \cup V_1 \setminus (V'_4 \cup V''_4)$, $V'_4 \cup V'_3$, and $(V_3 \setminus V'_3) \cup V''_4$.
		\end{itemize}
	\item
		When $|V'_4| > |V_1|$, we have $|V'_4| \le 2 |V_1|$ and $|V'_3| \le 2 |V_1|$;
		the operation creates three new parts $V_4 \cup V_1 \setminus V'_4$, $V'_4 \cup V'_3$, and $V_3 \setminus V'_3$.
	\item
		In all the above three cases of updating, the part $V_2$ is kept unchanged.
	\end{itemize}
\end{itemize}
\end{definition}

\begin{lemma}
\label{lemma08}
When there are multiple components of $G[V_4 \setminus \{u\}]$ adjacent to $V_3$ with their total size greater than $|V_1| + |V_2| + \frac {11}{24} |V_4|$ in Case 1,
and an operation {\em \mbox{Bridge-1}($V_3, V_4$)} updates the tetrapartition,
then the updated tetrapartition is feasible and better.
\end{lemma}
\begin{proof}
Recall that there are three cases of updating the tetrapartition.

In the first two cases, the operation achieves a subset $V'_3$ of size $|V'_3| > 2 |V_1|$,
and by the sub-collection minimality and each component of $G[V_3 \setminus \{v\}]$ being smaller than $\frac 13 |V_3|$, $|V'_3| < 2 |V_1| + \frac 13 |V_3|$.
In the first case, $V''_4$ can be located and again by sub-collection minimality we have $|V_1| < |V'_4| + |V''_4| \le 2 |V_1|$.
Therefore, for the three new parts,
$|V_4 \cup V_1 \setminus (V'_4 \cup V''_4)| < |V_4|$,
$|(V'_4 \cup V''_4) \cup V'_3| < 4 |V_1| + \frac 13 |V_3| < |V_4|$, and
$|V_3 \setminus V'_3| < |V_3|$, suggesting a better partition.

In the second case, the components of $G[V_4 \setminus \{u\}]$ that aren't adjacent to $V'_3$ have their total size exceeding $|V_2| + \frac {11}{24} |V_4|$,
and thus $V''_4$ can be located and again by sub-collection minimality we have $|V_1| < |V'_4| + |V''_4| \le 2 |V_1|$.
Therefore, for the three new parts,
$|V_4 \cup V_1 \setminus (V'_4 \cup V''_4)| < |V_4|$,
$|V'_4 \cup V'_3| < 3|V_1| + \frac 13 |V_3| < |V_4|$, and
$|(V_3 \setminus V'_3) \cup V''_4| < |V_3|$, suggesting a better partition.

The third case is similar to the first case.
In this case, the operation achieves a subset $V'_4$ of size $|V'_4| > |V_1|$,
and by the sub-collection minimality and each component of $G[V_4 \setminus \{u\}]$ being no larger than $|V_1|$, $|V'_4| \le 2 |V_1|$.
Therefore, for the three new parts,
$|V_4 \cup V_1 \setminus V'_4| < |V_4|$,
$|V'_4 \cup V'_3| \le 4 |V_1| < |V_4|$, and
$|V_3 \setminus V'_3| < |V_3|$, suggesting a better partition.

This proves the lemma.
\end{proof}

\begin{lemma}
\label{lemma09}
When there are multiple components of $G[V_4 \setminus \{u\}]$ adjacent to $V_3$ with their total size greater than $|V_1| + |V_2| + \frac {11}{24} |V_4|$ in Case 1,
no Bridge-1 operation is applicable,
every connected component of $G[V \setminus \{u, v\}]$ has size at most $\max\{3 |V_1|, |V_2|\} \le 2 |V_1| + |V_2|$.
\end{lemma}
\begin{proof}
From the definition of the Bridge-1 operation,
when it starts with a component $V'_4$ of $G[V_4 \setminus \{u\}]$ (or a component $V'_3$ of $G[V_3 \setminus \{v\}]$, respectively),
at the end it achieves $|V'_4| \le |V_1|$ and $|V'_3| \le 2 |V_1|$ without updating the partition.
Clearly, $G[V'_4 \cup V'_3]$ is a connected component of $G[V \setminus \{u, v\}]$.
One also sees that $V_2$ is also a connected component of $G[V \setminus \{u, v\}]$.
Therefore, every connected component of $G[V \setminus \{u, v\}]$ has size at most $\max\{3 |V_1|, |V_2|\} \le 2 |V_1| + |V_2|$.
\end{proof}

In the remaining case of Case 1 where the partition $\{V_1, V_2, V_3, V_4\}$ is yet unknown to be a $\frac {24}{13}$-approximation,
by Lemma~\ref{lemma09} we know that the graph $G = (V, E)$ exhibits a ``bi-star''-like configuration, with respect to the tetrapartition,
in that there is a vertex $u \in V_4$ ($v \in V_3$, respectively) such that $G[V_4 \setminus \{u\}]$ ($G[V_3 \setminus \{v\}]$, respectively) is disconnected,
and every connected component of $G[V \setminus \{u, v\}]$ has size at most $2 |V_1| + |V_2|$.
In an optimal tetrapartition denoted as $\{V^*_1, V^*_2, V^*_3, V^*_4\}$,
at least two parts contain none of the two center vertices $u$ and $v$, and thus their sizes are at most $2 |V_1| + |V_2|$.
Consequently $|V^*_4| \ge \frac 12 (|V| - 4 |V_1| - 2 |V_2|) \ge \frac 7{12} |V_4|$.
That is, the current tetrapartition $\{V_1, V_2, V_3, V_4\}$ is a $\frac {12}7$-approximation.
We conclude the following theorem:

\begin{theorem}
\label{thm04}
In Case 1, if there are multiple components of $G[V_4 \setminus \{u\}]$ adjacent to $V_3$ with their total size greater than $|V_1| + |V_2| + \frac {11}{24} |V_4|$
and the {\em Bridge-1} operation isn't applicable,
then the partition $\{V_1, V_2, V_3, V_4\}$ is a $\frac {12}7$-approximation.
\end{theorem}

In Case 2, none of $V_1$ and $V_2$ is adjacent to $V_4$ ({\it i.e.}, both $V_1$ and $V_2$ are adjacent to $V_3$ only and at the vertex $v \in V_3$ only).
One sees that Case 2 is almost symmetric to Case 1, by switching $V_3$ with $V_4$;
nevertheless, since $V_3$ might be strictly smaller than $V_4$, the argument differs slightly.

We have seen that $G[V_3]$ exhibits a nice star-like configuration (Lemma~\ref{lemma07}), but the connectivity configuration of $G[V_4]$ is unclear.
We next bipartition $V_4$ into $\{V_{41}, V_{42}\}$ as evenly as possible with $|V_{41}| \le |V_{42}|$.
If $|V_{42}| \le \frac 23 |V_4|$,
and assuming there are multiple components of $G[V_3 \setminus \{v\}]$ adjacent to $V_{4i}$ (for some $i \in \{1, 2\}$) with their total size greater than $|V_1|$,
then we find a minimal sub-collection of these components of $G[V_3 \setminus \{v\}]$ adjacent to $V_{4i}$ with their total size exceeding $|V_1|$,
denote by $V'_3$ the union of their vertex sets,
and subsequently create three new parts $V_3 \cup V_1 \setminus V'_3$, $V'_3 \cup V_{4i}$, and $V_{4,3-i}$,
while keeping $V_2$ unchanged.
One sees that this new tetrapartition is feasible and better, since $|V'_3| \le 2 |V_1| < \frac 13 |V_4|$.

In the other case, by Lemma~\ref{lemma03}, $G[V_4]$ also exhibits a nice star-like configuration centering at some vertex $u$,
such that $G[V_4 \setminus \{u\}]$ is disconnected and each component of $G[V_4 \setminus \{u\}]$ has size less than $\frac 13 |V_4|$.
See for an illustration in Figure~\ref{fig07}.
Furthermore, if the vertices of $V_3$ aren't adjacent to any vertex of $V_4$ other than $u$,
then very the same as in Case 1 the graph $G = (V, E)$ exhibits a ``bi-star''-like configuration, with respect to the partition,
in that there is a vertex $u \in V_4$ ($v \in V_3$, respectively) such that $G[V_4 \setminus \{u\}]$ ($G[V_3 \setminus \{v\}]$, respectively) is disconnected,
and every connected component of $G[V \setminus \{u, v\}]$ has size less than $\frac 13 |V_4|$.
The same succeeding argument states that the tetrapartition $\{V_1, V_2, V_3, V_4\}$ is a $\frac 32$-approximation.

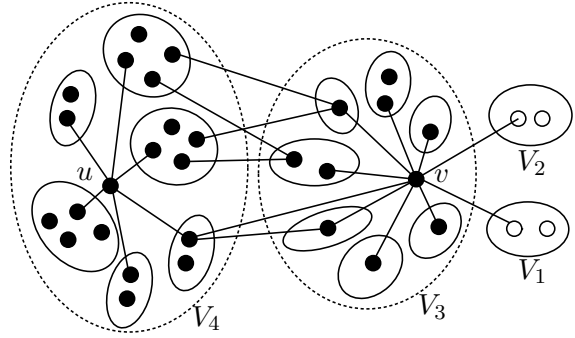
\begin{figure}[htpb]
\begin{center}
  \setlength{\unitlength}{0.7bp}%
  \begin{picture}(315.22, 188.75)(0,0)
  \put(0,0){\includegraphics[scale=0.7]{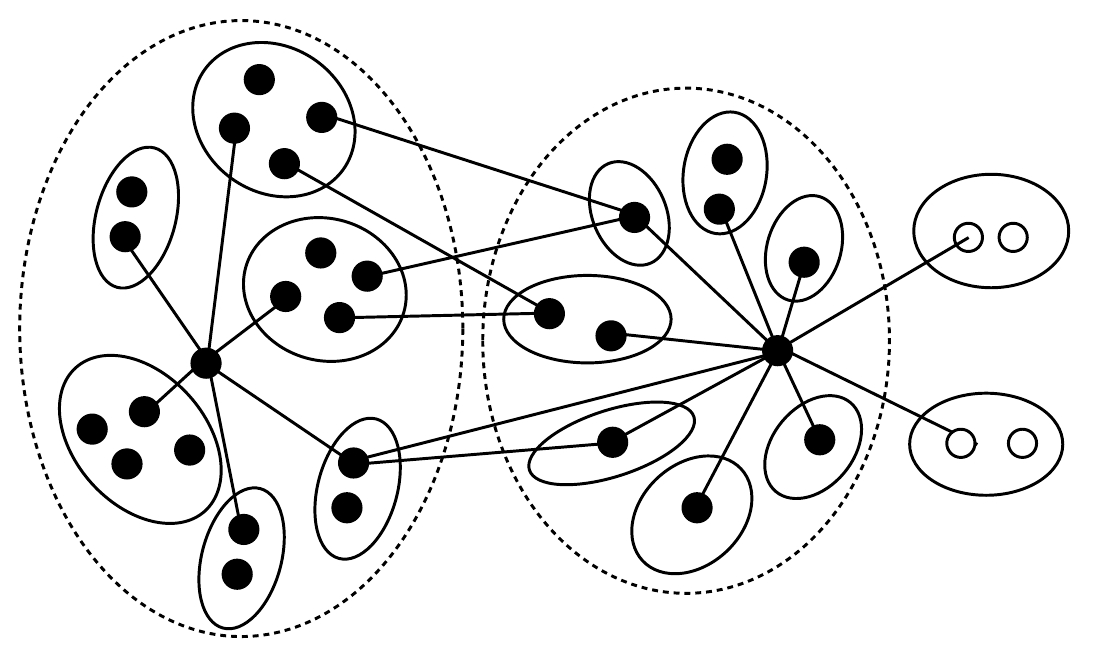}}
  \put(233.62,86.51){\fontsize{11.38}{13.66}\selectfont $v$}
  \put(224.76,15.16){\fontsize{11.38}{13.66}\selectfont $V_3$}
  \put(278.69,34.44){\fontsize{11.38}{13.66}\selectfont $V_1$}
  \put(278.35,92.76){\fontsize{11.38}{13.66}\selectfont $V_2$}
  \put(103.08,9.11){\fontsize{11.38}{13.66}\selectfont $V_4$}
  \put(41.06,87.52){\fontsize{11.38}{13.66}\selectfont $u$}
  \end{picture}%
\end{center}
\caption{An illustration of the ``bi-star''-like configuration of the graph $G = (V, E)$,
	with respect to the tetrapartition $\{V_1, V_2, V_3, V_4\}$ in Case 2.\label{fig07}}
\end{figure}

In the remaining case of Case 2, there are vertices of $V_3$ adjacent to some vertices of $V_4 \setminus \{u\}$,
and we design the following {\em Bridge-2} operation almost identical to Bridge-1, but the roles of $V_3$ and $V_4$ are swapped.

\begin{theorem}
\label{thm05}
In Case 2, let $V'_3$ denote the union of the vertex sets of all the components of $G[V_3 \setminus \{v\}]$ that are adjacent to $V_4$;
if $|V'_3| \le |V_2| + \frac {11}{24} |V_4|$, then the partition $\{V_1, V_2, V_3, V_4\}$ is a $\frac {24}{13}$-approximation.
\end{theorem}
\begin{proof}
See for an illustration of Case 2 in Figure~\ref{fig07}.
We note that the theorem statement is almost the same as Theorem~\ref{thm03},
but the quantity $|V_1|$ disappears since it is combined with $|V_3|$ to replace $|V_4|$.
The proof thus is almost the same as the proof for Theorem~\ref{thm03}.

If $|V'_3| \le |V_2| + \frac {11}{24} |V_4| < \frac 9{24} |V_4|$ (using Eq.~(\ref{eq01}) we have $|V_3| > \frac {20}{24} |V_4|$),
then denote the components of $G[V_3 \setminus \{v\}]$ not adjacent to $V_4$ as $G[V^v_{3i}]$, $i = 1, 2, \ldots \ell$, with $\ell \ge 3$.
Recall that $|V_1| + |V_3| \ge |V_4|$ and $|V^v_{3i}| \le |V_1|$, for each $i = 1, 2, \ldots \ell$.
In an optimal $4$-partition denoted as $\{V^*_1, V^*_2, V^*_3, V^*_4\}$, assume $v \in V^*_j$.
If $V^*_j$ contains all these $\ell + 2$ subsets, $V_1, V_2, V^v_{3i}, i = 1, 2, \ldots \ell$,
then $|V^*_j| = |V_1| + |V_2| + |V_3| - |V'_3| \ge \frac {13}{24} |V_4|$.
In the other case, at least one of these $\ell + 2$ subsets becomes a separate part in $\{V^*_1, V^*_2, V^*_3, V^*_4\}$,
of which the size is at most $|V_2|$,
and thus by item 3) of Lemma~\ref{lemma07} we have $|V^*_4| \ge \frac 13 (|V_1| + |V_3| + |V_4|) \ge \frac 23 |V_4|$.
Therefore, we always have $|V^*_4| \ge \frac {13}{24} |V_4|$, and thus the partition $\{V_1, V_2, V_3, V_4\}$ is a $\frac {24}{13}$-approximation.
\end{proof}

\begin{definition}
\label{def06}
Operation {\em Bridge-2($V_3, V_4$):} 
\begin{itemize}
\parskip=0pt
\item
	{\em precondition:} In Case 2,
	there are multiple components of $G[V_3 \setminus \{v\}]$ adjacent to $V_4$ with their total size greater than $|V_2| + \frac {11}{24} |V_4|$,
	and there is a vertex $u \in V_4$ such that $G[V_4 \setminus \{u\}]$ is disconnected and each component has size less than $\frac 13 |V_4|$.
\item
	{\em effect:} Find a component $G[V^u_{4x}]$ of $G[V_4 \setminus \{u\}]$ (could be empty)
	that is adjacent to a component $G[V^v_{3y}]$ of $G[V_3 \setminus \{v\}]$;
	initialize $V'_4$ to be $V^u_{4x}$ and $V'_3$ to be $V^v_{3y}$;
	iteratively,
	\begin{itemize}
	\parskip=0pt
	\item
		let ${\cal C}_3$ denote the collection of the components of $G[V_3 \setminus \{v\}]$ that are adjacent to $V'_4$, excluding $V'_3$;
		\begin{itemize}
		\parskip=0pt
		\item
		if the total size of components in ${\cal C}_3$ exceeds $|V_1| - |V'_3|$,
		then the operation greedily finds a minimal sub-collection of these components of ${\cal C}_3$ with their total size exceeding $|V_1| - |V'_3|$,
		adds their vertex sets to $V'_3$,
		and proceeds to termination;
		\item
		if the total size of components in ${\cal C}_3$ is less than $|V_1| - |V'_3|$,
		then the operation adds the vertex sets of all these components to $V'_3$;
	\end{itemize}
	\item
		let ${\cal C}_4$ denote the collection of the components of $G[V_4 \setminus \{u\}]$ that are adjacent to $V'_3$, excluding $V'_4$;
		\begin{itemize}
		\parskip=0pt
		\item
		if the total size of components in ${\cal C}_4$ exceeds $2 |V_1| - |V'_4|$,
		then the operation greedily finds a minimal sub-collection of these components of ${\cal C}_4$ with their total size exceeding $2 |V_1| - |V'_4|$,
		adds their vertex sets to $V'_4$,
		and proceeds to termination;
		\item
		if the total size of components in ${\cal C}_4$ is less than $2 |V_1| - |V'_4|$,
		then the operation adds the vertex sets of all these components to $V'_4$;
		\end{itemize}
	\item
		if both ${\cal C}_3$ and ${\cal C}_4$ are empty, then the operation terminates without updating the partition.
	\end{itemize}
	At termination, exactly one of $|V'_4| > 2 |V_1|$ and $|V'_3| > |V_1|$ holds. 
	\begin{itemize}
	\parskip=0pt
	\item
		When $|V'_4| > 2 |V_1|$, we have $|V'_3| \le |V_1|$ and $|V'_4| < 2 |V_1| + \frac 13 |V_4|$;
		\begin{itemize}
		\parskip=0pt
		\item
			if the collection of the components of $G[V_3 \setminus \{v\}]$ that are adjacent to $V'_4$, excluding $V'_3$, exceeds $|V_1| - |V'_3|$,
			then the operation greedily finds a minimal sub-collection of these components with their total size exceeding $|V_1| - |V'_3|$,
			and denotes by $V''_3$ the union of their vertex sets;
			subsequently, the operation creates three new parts $V_3 \cup V_1 \setminus (V'_3 \cup V''_3)$, $(V'_3 \cup V''_3) \cup V'_4$, and $V_4 \setminus V'_4$;
		\item
			otherwise, the operation greedily finds a minimal sub-collection of the components of $G[V_3 \setminus \{v\}]$ that aren't adjacent to $V'_4$
			with their total size exceeding $|V_1| - |V'_3|$,
			and denotes by $V''_3$ the union of their vertex sets;
			subsequently, the operation creates three new parts $V_3 \cup V_1 \setminus (V'_3 \cup V''_3)$, $V'_3 \cup V'_4$, and $(V_4 \setminus V'_4) \cup V''_3$.
		\end{itemize}
	\item
		When $|V'_3| > |V_1|$, we have $|V'_3| \le 2 |V_1|$ and $|V'_4| \le 2 |V_1|$;
		the operation creates three new parts $V_3 \cup V_1 \setminus V'_3$, $V'_3 \cup V'_4$, and $V_4 \setminus V'_4$.
	\item
		In all the above three cases of updating, the part $V_2$ is kept unchanged.
	\end{itemize}
\end{itemize}
\end{definition}

\begin{lemma}
\label{lemma10}
When there are multiple components of $G[V_3 \setminus \{v\}]$ adjacent to $V_4$ with their total size greater than $|V_2| + \frac {11}{24} |V_4|$ in Case 2,
and an operation {\em \mbox{Bridge-2}($V_3, V_4$)} updates the tetrapartition,
then the updated partition is feasible and better.
\end{lemma}
\begin{proof}
The proof is almost identical to the proof of Lemma~\ref{lemma08}, by using $4 |V_1| < \frac 23 |V_4|$.
\end{proof}

The following lemma states the same property as in Lemma~\ref{lemma09}, and we conclude the same Theorem~\ref{thm06}.

\begin{lemma}
\label{lemma11}
When there are multiple components of $G[V_3 \setminus \{v\}]$ adjacent to $V_4$ with their total size greater than $|V_2| + \frac {11}{24} |V_4|$ in Case 2,
no Bridge-2 operation is applicable,
every connected component of $G[V_4 \cup V_3 \setminus \{u, v\}]$ has size at most $\max\{3 |V_1|, |V_2|\}$.
\end{lemma}

\begin{theorem}
\label{thm06}
In Case 2, if there are multiple components of $G[V_3 \setminus \{v\}]$ adjacent to a vertex $u \in V_4$ with their total size greater than
$|V_2| + \frac {11}{24} |V_4|$ and no {\em Bridge-2} operation is applicable,
then the tetrapartition $\{V_1, V_2, V_3, V_4\}$ is a $\frac {12}7$-approximation.
\end{theorem}

Case 3 is different from the above two cases, as one of $V_1$ and $V_2$ is adjacent to $V_3$ while the other adjacent to $V_4$.
In fact, by Lemma~\ref{lemma07}, the graph $G = (V, E)$ already exhibits a ``bi-star''-like configuration, with respect to the partition,
in that there is a vertex $u \in V_4$ ($v \in V_3$, respectively) such that $G[V_4 \setminus \{u\}]$ ($G[V_3 \setminus \{v\}]$, respectively) is disconnected,
and every connected component of $G[V \setminus \{u, v\}]$ has size at most $|V_2|$.
The argument thus can be made slightly simpler.

\begin{theorem}
\label{thm07}
In Case 3, assume $V_i \in \{V_1, V_2\}$ is adjacent to $V_3$ at the vertex $v \in V_3$ while $V_j$ ($j = 3-i$) is adjacent to the vertex $u \in V_4$,
\begin{itemize}
\item
	let $V'_3$ denote the union of the vertex sets of all the components of $G[(V_3 \cup V_i) \setminus \{v\}]$ that are adjacent to $V_4$;
	if $|V'_3| \le |V_i| + \frac 7{24} |V_4|$, then the partition $\{V_1, V_2, V_3, V_4\}$ is a $\frac {24}{13}$-approximation;
\item
	let $V'_4$ denote the union of the vertex sets of all the components of $G[(V_4 \cup V_j) \setminus \{u\}]$ that are adjacent to $V_3$;
	if $|V'_4| \le |V_j| + \frac {11}{24} |V_4|$, then the partition $\{V_1, V_2, V_3, V_4\}$ is a $\frac {24}{13}$-approximation.
\end{itemize}
\end{theorem}
\begin{proof}
From Lemma~\ref{lemma07} and Eq.~(\ref{eq01}) we know that $|V_i| < \frac 16 |V_4|$ and thus $|V_3| > \frac 56 |V_4|$.
For ease of presentation, using Lemma~\ref{lemma07} we regard $G[V_i]$ as a component of $G[(V_3 \cup V_i) \setminus \{v\}]$,
treated the same as the components of $G[V_3 \setminus \{v\}]$.

If $|V'_3| \le |V_i| + \frac 7{24} |V_4| < \frac {11}{24} |V_4|$,
then denote the components of $G[(V_3 \cup V_i) \setminus \{v\}]$ that are not adjacent to $V_4$ as $G[V^v_{3x}]$, $x = 1, 2, \ldots \ell$, with $\ell \ge 3$.
Recall that $|V^v_{3x}| \le |V_i|$, for each $x = 1, 2, \ldots \ell$.
In an optimal $4$-partition denoted as $\{V^*_1, V^*_2, V^*_3, V^*_4\}$, assume $v \in V^*_y$.
If $V^*_y$ contains all these $\ell$ subsets, then $|V^*_y| = |V_i| + |V_3| - |V'_3| > \frac {13}{24} |V_4|$.
In the other case, at least one of these $\ell$ subsets becomes a separate part in $\{V^*_1, V^*_2, V^*_3, V^*_4\}$,
of which the size is at most $|V_i|$,
and thus we have $|V^*_4| \ge \frac 13 (|V_j| + |V_3| + |V_4|) \ge \frac {11}{18} |V_4|$.
Therefore, we always have $|V^*_4| \ge \frac {13}{24} |V_4|$, and thus the partition $\{V_1, V_2, V_3, V_4\}$ is a $\frac {24}{13}$-approximation.

Also from Eq.~(\ref{eq01}) we have $|V_j| < \frac 16 |V_4|$.
Similarly, for ease of presentation, using Lemma~\ref{lemma07} we regard $G[V_j]$ as a component of $G[(V_4 \cup V_j) \setminus \{u\}]$,
treated the same as the components of $G[V_4 \setminus \{u\}]$.

If $|V'_4| \le |V_j| + \frac {11}{24} |V_4| < \frac {15}{24} |V_4|$, then there are components of $G[(V_4 \cup V_j) \setminus \{u\}]$ not adjacent to $V_3$,
denoted as $G[V^u_{4x}]$, $x = 1, 2, \ldots \ell$, with $\ell \ge 3$.
Recall that $|V^u_{4x}| \le |V_j|$, for each $x = 1, 2, \ldots \ell$.
In an optimal $4$-partition denoted as $\{V^*_1, V^*_2, V^*_3, V^*_4\}$, assume $u \in V^*_y$.
If $V^*_y$ contains all these $\ell$ subsets, then $|V^*_y| = |V_j| + |V_4| - |V'_4| \ge \frac {13}{24} |V_4|$.
In the other case, at least one of these $\ell$ subsets becomes a separate part in $\{V^*_1, V^*_2, V^*_3, V^*_4\}$,
of which the size is at most $|V_j|$,
and thus by item 3) of Lemma~\ref{lemma07} we have $|V^*_4| \ge \frac 13 (|V_i| + |V_3| + |V_4|) \ge \frac 23 |V_4|$.
Therefore, we always have $|V^*_4| \ge \frac {13}{24} |V_4|$, and thus the partition $\{V_1, V_2, V_3, V_4\}$ is a $\frac {24}{13}$-approximation.
\end{proof}

\begin{corollary}
\label{coro01}
In Case 3, assume $V_i \in \{V_1, V_2\}$ is adjacent to $V_3$ at the vertex $v \in V_3$ while $V_j$ ($j = 3-i$) is adjacent to the vertex $u \in V_4$,
if the tetrapartition $\{V_1, V_2, V_3, V_4\}$ is not yet a $\frac {24}{13}$-approximation,
then $|V_3 \cup V_1| \ge |V_4|$ (which is slightly stronger than $|V_3 \cup V_i| \ge |V_4|$ stated in Lemma~\ref{lemma07}).
\end{corollary}
\begin{proof}
When $i = 1$, from Lemma~\ref{lemma07} we have $|V_3 \cup V_1| \ge |V_4|$.
When $i = 2$, then $j = 1$ and thus from Lemma~\ref{lemma07} every component of $G[V_4 \setminus \{u\}]$ has size at most $|V_1|$.
Theorem~\ref{thm07} says that some component of $G[V_4 \setminus \{u\}]$, say $G[V^u_{4x}]$, is adjacent to $V_3$.
Since Pull($V^u_{4x} \subset V_4, V_3$) is not applicable, we conclude that $|V_3 \cup V_1| \ge |V_3 \cup V^u_{4x}| \ge |V_4|$.
Therefore, we always have $|V_3 \cup V_1| \ge |V_4|$.
This proves the corollary.
\end{proof}

\begin{figure}[htpb]
\begin{center}
  \setlength{\unitlength}{0.7bp}%
  \begin{picture}(362.06, 188.75)(0,0)
  \put(0,0){\includegraphics[scale=0.7]{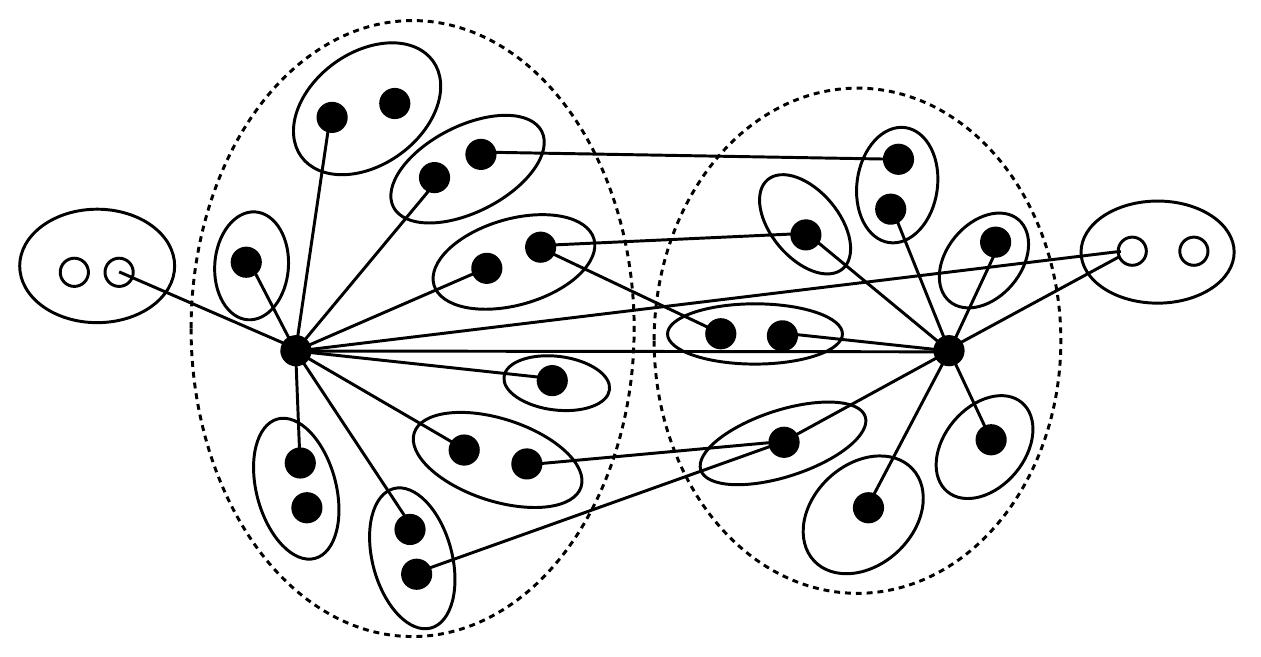}}
  \put(277.95,81.48){\fontsize{11.38}{13.66}\selectfont $v$}
  \put(274.11,15.16){\fontsize{11.38}{13.66}\selectfont $V_3$}
  \put(328.05,89.78){\fontsize{11.38}{13.66}\selectfont $V_i$}
  \put(20.86,82.70){\fontsize{11.38}{13.66}\selectfont $V_j$}
  \put(67.36,11.01){\fontsize{11.38}{13.66}\selectfont $V_4$}
  \put(69.75,78.18){\fontsize{11.38}{13.66}\selectfont $u$}
  \end{picture}%
\end{center}
\caption{An illustration of the ``bi-star''-like configuration of the graph $G = (V, E)$,
	with respect to the tetrapartition $\{V_1, V_2, V_3, V_4\}$ in Case 3,
	where $V_i \in \{V_1, V_2\}$ is adjacent to $V_3$ at the vertex $v \in V_3$ and $V_j$ ($j = 3-i$) is adjacent to the vertex $u \in V_4$.\label{fig08}}
\end{figure}
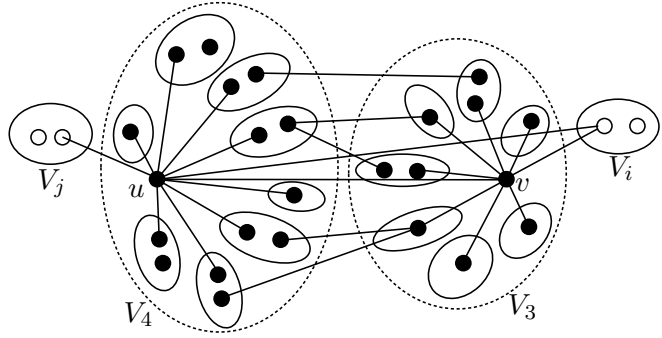

We assume that $V_i \in \{V_1, V_2\}$ is adjacent to $V_3$ at the vertex $v \in V_3$ while $V_j$ ($j = 3-i$) is adjacent to the vertex $u \in V_4$ in Case 3.
Recall from item 6) of Lemma~\ref{lemma07} that if $V_i$ ($V_j$, respectively) is also adjacent to $V_4$ ($V_3$, respectively),
then it is adjacent to the vertex $u$ ($v$, respectively).
See for an illustration of Case 3 in Figure~\ref{fig08}.
The following operation Bridge-3($V_3, V_4$) is again almost identical to Bridge-1 and Bridge-2 operations,
but slightly simpler.

\begin{definition}
\label{def07}
Operation {\em Bridge-3($V_3, V_4$):} 
\begin{itemize}
\parskip=0pt
\item
	{\em precondition:}  In Case 3,
	assume $V_i \in \{V_1, V_2\}$ is adjacent to $V_3$ at the vertex $v \in V_3$ while $V_j$ ($j = 3-i$) is adjacent to the vertex $u \in V_4$,
	there are multiple components of $G[(V_3 \cup V_i) \setminus \{v\}]$ adjacent to $V_4$ with their total size greater than $|V_i| + \frac 7{24} |V_4|$,
	and there are multiple components of $G[(V_4 \cup V_j) \setminus \{u\}]$ adjacent to $V_3$ with their total size greater than $|V_j| + \frac {11}{24} |V_4|$;
\item
	{\em effect:} Find a component $G[V^u_{4x}]$ of $G[V_4 \setminus \{u\}]$ (could be empty)
	that is adjacent to a component $G[V^v_{3y}]$ of $G[V_3 \setminus \{v\}]$;
	initialize $V'_4$ to be $V^u_{4x}$ and $V'_3$ to be $V^v_{3y}$;
	iteratively,
	\begin{itemize}
	\parskip=0pt
	\item
		let ${\cal C}_3$ denote the collection of the components of $G[V_3 \setminus \{v\}]$ that are adjacent to $V'_4$, excluding $V'_3$;
		\begin{itemize}
		\parskip=0pt
		\item
		if the total size of components in ${\cal C}_3$ exceeds $|V_i| - |V'_3|$,
		then the operation greedily finds a minimal sub-collection of these components of ${\cal C}_3$ with their total size exceeding $|V_i| - |V'_3|$,
		adds their vertex sets to $V'_3$,
		and proceeds to termination;
		\item
		if the total size of components in ${\cal C}_3$ is less than $|V_i| - |V'_3|$,
		then the operation adds the vertex sets of all these components to $V'_3$;
	\end{itemize}
	\item
		let ${\cal C}_4$ denote the collection of the components of $G[V_4 \setminus \{u\}]$ that are adjacent to $V'_3$, excluding $V'_4$;
		\begin{itemize}
		\parskip=0pt
		\item
		if the total size of components in ${\cal C}_4$ exceeds $|V_j| - |V'_4|$,
		then the operation greedily finds a minimal sub-collection of these components of ${\cal C}_4$ with their total size exceeding $|V_j| - |V'_4|$,
		adds their vertex sets to $V'_4$,
		and proceeds to termination;
		\item
		if the total size of components in ${\cal C}_4$ is less than $|V_j| - |V'_4|$,
		then the operation adds the vertex sets of all these components to $V'_4$;
		\end{itemize}
	\item
		if both ${\cal C}_3$ and ${\cal C}_4$ are empty, then the operation proceeds to termination.
	\end{itemize}
	At termination, exactly one of $|V'_3| > |V_i|$ and $|V'_4| > |V_j|$ holds. 
	\begin{itemize}
	\parskip=0pt
	\item
		When $|V'_3| > |V_i|$, we have $|V'_4| \le |V_j|$ and $|V'_3| \le 2 |V_i|$;
		the operation creates three new parts $V'_4 \cup V'_3$, $V_4 \setminus V'_4$, and $(V_3 \setminus V'_3) \cup V_i$,
		while keeping the part $V_j$ unchanged.
	\item
		When $|V'_4| > |V_j|$, we have $|V'_3| \le |V_i|$ and $|V'_4| \le 2 |V_j|$;
		the operation creates three new parts $V'_4 \cup V'_3$, $(V_4 \setminus V'_4) \cup V_j$, and $V_3 \setminus V'_3$,
		while keeping the part $V_i$ unchanged.
	\end{itemize}
\end{itemize}
\end{definition}

\begin{lemma}
\label{lemma12}
In Case 3, assume $V_i \in \{V_1, V_2\}$ is adjacent to $V_3$ at the vertex $v \in V_3$ while $V_j$ ($j = 3-i$) is adjacent to the vertex $u \in V_4$,
there are multiple components of $G[(V_3 \cup V_i) \setminus \{v\}]$ adjacent to $V_4$ with their total size greater than $|V_i| + \frac 7{24} |V_4|$,
there are multiple components of $G[(V_4 \cup V_j) \setminus \{u\}]$ adjacent to $V_3$ with their total size greater than $|V_j| + \frac {11}{24} |V_4|$,
and a {\em Bridge-3($V_3, V_4$)} operation updates the tetrapartition,
then the updated partition is feasible and better.
\end{lemma}
\begin{proof}
Recall from Eq.~(\ref{eq01}) that $|V_1| \le |V_2| < \frac 16 |V_4|$.
When the Bridge-3($V_3, V_4$) operation achieves a pair $(V'_3, V'_4)$,
by the minimality of the sub-collection we know that $|V'_3| \le 2 |V_i|$ and $|V'_4| \le 2 |V_j|$,
suggesting $|V'_3 \cup V'_4| < 3 \times \frac 16 |V_4| = \frac 12 |V_4|$.
Also, if $|V'_3| > |V_i|$, then for the other two new parts,
$|V_4 \setminus V'_4| < |V_4|$,
and $|(V_3 \setminus V'_3) \cup V_i| < |V_3|$;
if $|V'_4| > |V_j|$, then for the other two new parts,
$|(V_4 \setminus V'_4) \cup V_j| < |V_4|$,
and $|V_3 \setminus V'_3| < |V_3|$.
That is, the size of its largest part reduces by at least $1$.
This proves the lemma.
\end{proof}

The following lemma states the same property as in Lemmas~\ref{lemma09} and \ref{lemma11}, and we conclude the same Theorem~\ref{thm08}.

\begin{lemma}
\label{lemma13}
Assume $V_i \in \{V_1, V_2\}$ is adjacent to $V_3$ at the vertex $v \in V_3$ while $V_j$ ($j = 3-i$) is adjacent to the vertex $u \in V_4$ in Case 3.
When there are multiple components of $G[(V_3 \cup V_i) \setminus \{v\}]$ adjacent to $V_4$ with their total size greater than $|V_i| + \frac 7{24} |V_4|$,
there are multiple components of $G[(V_4 \cup V_j) \setminus \{u\}]$ adjacent to $V_3$ with their total size greater than $|V_j| + \frac {11}{24} |V_4|$,
and no Bridge-3 operation is applicable,
every connected component of $G[V_4 \cup V_3 \setminus \{u, v\}]$ has size at most $|V_1| + |V_2|$.
\end{lemma}

\begin{theorem}
\label{thm08}
In Case 3, assume $V_i \in \{V_1, V_2\}$ is adjacent to $V_3$ at the vertex $v \in V_3$ while $V_j$ ($j = 3-i$) is adjacent to the vertex $u \in V_4$,
and no {\em Bridge-3} operation is applicable,
then the partition $\{V_1, V_2, V_3, V_4\}$ is a $\frac 43$-approximation.
\end{theorem}

Combining all the three cases,
we can design the following algorithm {\sc Approx-$4$} as depicted in Figure~\ref{fig09} for the {\sc $4$-BGP} problem,
which is iterative in nature and in every iteration it applies one of the Merge and the Pull and the three Bridge operations.
And we have the following final conclusion for the {\sc $4$-BGP} problem:

\begin{figure}[ht]
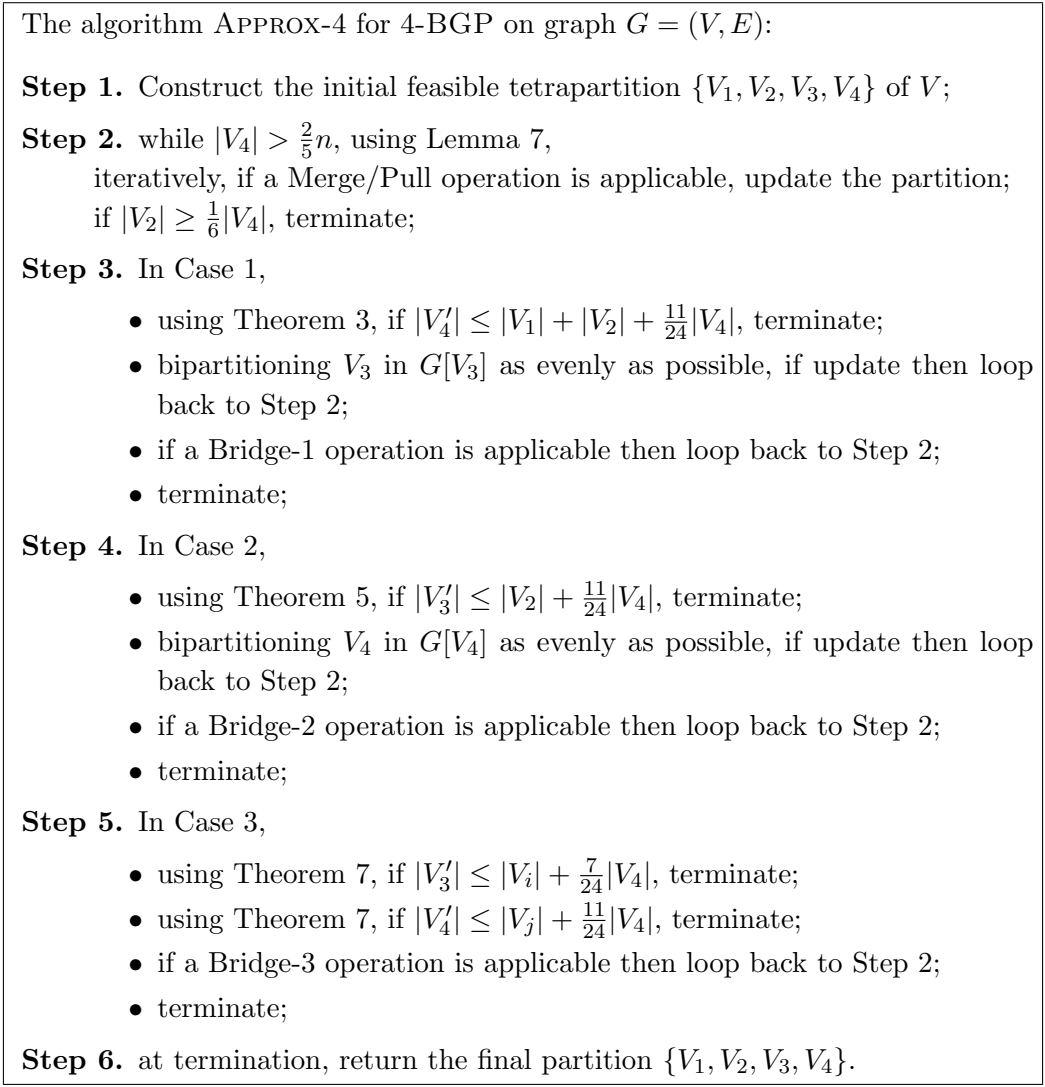

\begin{center}
\framebox{
\begin{minipage}{5.3in}
	The algorithm {\sc Approx-$4$} for {\sc $4$-BGP} on graph $G = (V, E)$:
	\begin{description}
	\parskip=0pt
	\item[Step 1.]
		Construct the initial feasible tetrapartition $\{V_1, V_2, V_3, V_4\}$ of $V$;
	\item[Step 2.]
		while $|V_4| > \frac 25 n$, using Lemma~\ref{lemma07},\\
		\hspace{1in} iteratively, if a Merge/Pull operation is applicable, update the partition;\\
		\hspace{1in} if $|V_2| \ge \frac 16 |V_4|$, terminate;
	\item[Step 3.]
		In Case 1,
		\begin{itemize}
		\parskip=0pt
		\item
			using Theorem~\ref{thm03}, if $|V'_4| \le |V_1| + |V_2| + \frac {11}{24} |V_4|$, terminate;
		\item
			bipartitioning $V_3$ in $G[V_3]$ as evenly as possible, if update then loop back to Step 2;
		\item
			if a Bridge-1 operation is applicable then loop back to Step 2;
		\item
			terminate;
		\end{itemize}
	\item[Step 4.]
		In Case 2,
		\begin{itemize}
		\parskip=0pt
		\item
			using Theorem~\ref{thm05}, if $|V'_3| \le |V_2| + \frac {11}{24} |V_4|$, terminate;
		\item
			bipartitioning $V_4$ in $G[V_4]$ as evenly as possible, if update then loop back to Step 2;
		\item
			if a Bridge-2 operation is applicable then loop back to Step 2;
		\item
			terminate;
		\end{itemize}
	\item[Step 5.]
		In Case 3,
		\begin{itemize}
		\parskip=0pt
		\item
			using Theorem~\ref{thm07}, if $|V'_3| \le |V_i| + \frac 7{24} |V_4|$, terminate;
		\item
			using Theorem~\ref{thm07}, if $|V'_4| \le |V_j| + \frac {11}{24} |V_4|$, terminate;
		\item
			if a Bridge-3 operation is applicable then loop back to Step 2;
		\item
			terminate;
		\end{itemize}
	\item[Step 6.]
		at termination, return the final partition $\{V_1, V_2, V_3, V_4\}$.
	\end{description}
\end{minipage}}
\end{center}
\caption{A high-level description of the algorithm {\sc Approx-$4$} for {\sc $4$-BGP}.\label{fig09}}
\end{figure}

\begin{theorem}
\label{thm09}
The algorithm {\sc Approx-$4$} is an $O(|V|^2 |E|)$-time $\frac {24}{13}$-approximation for the {\sc $4$-BGP} problem.
\end{theorem}
\begin{proof}
Note that every local improvement operation can be executed in $O(|V|+|E|)$ time, via a graph traversal.
Since each operation achieves a better partition, that is, either the size of the largest part is reduced by at least $1$,
or the largest part is unchanged by the second largest part is reduced by at least $1$,
the total number of executed operations is in $O(|V|^2)$.
We conclude that the running time of the algorithm {\sc Approx-$4$} is in $O(|V|^2 |E|)$.
The performance ratio is taken as the maximum among $\{\frac {24}{13}, \frac {12}7, \frac 32, \frac 43\}$, which is $\frac {24}{13}$.
\end{proof}

\section{Conclusions}
We studied the {\sc $k$-BGP} problem to partition the vertex set of a given simple connected graph $G = (V, E)$ into $k$ parts,
such that the subgraph induced by each part is connected and the maximum cardinality of these $k$ parts is minimized.
The problem is NP-hard, and approximation algorithms were proposed for only $k = 2, 3$.
We focus on $k \ge 4$,
and present a $k/2$-approximation algorithm for {\sc $k$-BGP}, for any fixed $k \ge 3$, and an improved $24/13$-approximation for {\sc $4$-BGP}.
Along the way, we have designed several intuitive and interesting local improvement operations.

There is no any non-trivial lower bound on the approximation ratio for the {\sc $k$-BGP} problem,
except $6/5$ for the problem when $k$ is part of the input.
We feel that it could be challenging to design better approximation algorithms for {\sc $2$-BGP} and {\sc $3$-BGP};
but for {\sc $4$-BGP} we believe the parameters in the three Bridge operations can be adjusted better, though non-trivially,
leading to an $8/5$-approximation.
We leave it open on whether or not {\sc $k$-BGP} admits an $o(k)$-approximation.

\subsubsection*{Acknowledgements.}
CY and AZ are supported by the NSFC Grants 11971139, 11771114 and 11571252;
they are also supported by the CSC Grants 201508330054 and 201908330090, respectively.
ZZC is supported by in part by the Grant-in-Aid for Scientific Research of the Ministry of Education, Science, Sports and Culture of Japan, under Grant No. 18K11183.
GL is supported by the NSERC Canada.


\end{document}